\documentclass[11pt]{article}

\usepackage[utf8]{inputenc} % allow utf-8 input
\usepackage[margin=1in]{geometry}
\usepackage[T1]{fontenc}    % use 8-bit T1 fonts
\usepackage{hyperref}       % hyperlinks
\usepackage{url}            % simple URL typesetting
\usepackage{booktabs}       % professional-quality tables
\usepackage{amsfonts}       % blackboard math symbols
\usepackage{nicefrac}       % compact symbols for 1/2, etc.
\usepackage{microtype}      % microtypography

\usepackage{algorithm,algorithmic}
\usepackage{graphicx}
\usepackage{amsmath}
\usepackage{amsfonts}
\usepackage{amssymb, bm}
\usepackage{amsthm}
\usepackage{float}
\usepackage{color}
\usepackage{subfigure}
\usepackage{rotating}
\usepackage{url}

\numberwithin{equation}{section}
\theoremstyle{plain}
\newtheorem{thm}{Theorem}[section]
\newtheorem{lem}{Lemma}[section]
\newtheorem{cor}{Corollary}[section]
\newtheorem{prop}{Proposition}[section]

\def \A {\mathcal{A}}
\def \R {\mathbb{R}}

\def \B {\mathcal{B}}

\def \x {\bm{x}}
\def \bt {\bm{\theta}}
\def \bT {\bm{\Theta}}
\def \bl {\bm{\lambda}}
\def \H {\mathbb{H}}

\def \D {\mathcal{D}}

\def \P {\mathbb{P}}
\def \E {\mathbb{E}}
\def \V {\mathbb{V}}

\begin{document}

%\nocite{*}

\title{Local Uncertainty Sampling for Large-Scale Multi-Class \\
Logistic Regression}

\author{Lei Han$^{*\ddagger}$, Kean Ming Tan$^{\dagger}$, Ting Yang$^{\ddagger}$, Tong Zhang$^{*}$
\bigskip\\
$^*$Tencent AI Lab, Shenzhen, China\\
$^{\dagger}$School of Statistics, University of Minnesota, USA\\
$^{\ddagger}$Department of Statistics, Rutgers University, USA
}

\date{}

\maketitle

\begin{abstract} 
A major challenge for building statistical models in the big data era is that the available data volume far
exceeds the computational capability.
A common approach for solving this problem is to employ a subsampled dataset that can be handled
by available computational resources.
In this paper, we propose a general subsampling scheme for large-scale multi-class logistic regression and examine the variance of the resulting estimator.
We show that asymptotically, the proposed method always achieves a smaller variance than that of the uniform random sampling. Moreover, when the classes are conditionally imbalanced, significant improvement over uniform sampling can be achieved.
Empirical performance of the proposed method is compared to other methods on both simulated and real-world datasets, and these results match and confirm our theoretical analysis.
\end{abstract}

\section{Introduction}\label{sec:intro}

In recent years, data volume has grown exponentially, and this has created demands for building statistical methods to analyze huge datasets. 
In particular, the size of these datasets far exceeds the available computational capability at hand.
For instance, it may not be computationally feasible to perform standard statistical procedures on a single machine when the datasets are huge.
Although one remedy is to develop sophisticated distributed computing systems that can directly handle large datasets, the increased
system complexity makes this approach not suitable for all scenarios.
Another remedy to this problem is to employ a subsampled dataset that can be handled by existing computational resources.
In fact, such an approach is widely used to solve big data problems. 
However, subsampling may suffer from loss of statistical accuracy, i.e., the variance of the resulting estimator may be large.
Therefore a natural question is to tradeoff statistical accuracy for  computational efficiency by
designing an effective sampling scheme that can minimize the reduction of statistical accuracy given
a certain computational capacity.

In this paper, we examine the subsampling approach for solving large-scale multi-class logistic regression problems that are common in practical applications.
The general idea of subsampling is to assign an acceptance probability for each data point and select observations according to the assigned probabilities.
After the subsampling procedure, only a small portion of the data are selected from the full dataset.  
Hence, the model built using the subsampled data will not be as accurate as that of the full data.
The key challenge is to design a good sampling scheme together with the corresponding estimation procedure 
such that the loss of statistical accuracy is minimized, given some fixed computational resource.  Here, the required computational resource can be measured by the number of subsampled data.

There has been substantial work on subsampling methods for large-scale statistical estimation problems \cite{cortes2008sample,cortes2010learning,dhillon2013new,mineiro2013loss,xie1989logit,zadrozny2004learning,zhang2000value}. The simplest method is to subsample the data uniformly. However, uniform subsampling assigns the same acceptance probability to every data point, which fails to differentiate the importance among the samples.
For example, a particular scenario often encountered in practical applications of logistic regression is when the class labels are imbalanced.
This problem has attracted significant interests in the machine learning literature (see survey papers in \cite{chawla2004editorial,he2009learning}). Generally, there are two types of commonly encountered class imbalance situations:  \emph{marginal imbalance} and \emph{conditional imbalance}. In the case of marginal imbalance, some classes are much rarer than the other classes. 
This situation often occurs in applications such as fraud and intrusion detection \cite{abe2004iterative, kim2002pattern}, disease diagnoses \cite{widodo2007support}, protein fold classification \cite{tan2003multi}, and etc.
On the other hand, conditional imbalance describes the case when the labels (denoted as $y$) for most observations (denoted as $\x$) are easy to predict.
This happens in applications with very accurate classifiers  such as handwriting digits recognition \cite{lecun1998gradient} and web/email spam filtering \cite{fithian2014local,webb2006introducing}.
Note that marginally imbalance implies conditional imbalance, while the reverse is not necessarily true.

For marginally imbalanced binary classification problems, case-control subsampling (CC), which  uniformly selects an equal number of samples from each class,
has been widely used in practice in epidemiology and social science studies \cite{mantel1959statistical}.
Under this scheme, equal number of samples are subsampled from each class, and therefore the sampled data are marginally balanced.
It is known that case-control subsampling is more efficient than uniform subsampling when the datasets are marginally imbalanced.
However, since the acceptance probability relies on the response variable in CC subsampling,
the distribution of subsampled data is skewed by the sample selection process \cite{breslow1982design}.
It follows that correction methods are needed to adjust for the selection bias \cite{anderson1972separate,king2001logistic}.
Another method to remove bias in CC subsampling is to weight each sampled data point by the inverse of its acceptance probability. This is known as the weighted case-control method, which has been shown to be consistent and unbiased \cite{horvitz1952generalization}, but may increase the variance of the resulting estimator \cite{scott1986fitting, scott1991fitting, scott2002robustness}.

One drawback of the standard case-control subsampling is that it  does not consider the case when the data are conditionally imbalanced.
This issue was addressed in \cite{fithian2014local}, who proposed an improved subsampling scheme called
{\em Local Case-Control} (LCC) sampling for binary logistic regression.
The LCC method assigns each data point an acceptance probability determined not only by its label but also by its observed covariates. 
LCC assigns more importance on data points that are easy to be mis-classified according to a consistent pilot estimator, which is an approximate conditional probability estimator possibly obtained using a small number of uniformly sampled data. The method in \cite{fithian2014local} fits a logistic model with the LCC sampled data, and then apply a post-estimation correction to the resulting estimator using the pilot estimator. Therefore, the LCC sampling approach belongs to the correction based methods such as that of \cite{anderson1972separate,king2001logistic}. It was shown in \cite{fithian2014local} that given a consistent pilot estimator,  the LCC estimator is consistent with an asymptotic variance that may significantly outperform that of the uniform sampling and CC based sampling methods when the data is strongly conditionally imbalanced.

In this paper, we propose an effective sampling strategy for large-scale multi-class logistic regression problems that generalizes LCC sampling. Our proposed sampling procedure can be summarized in the following:
\begin{itemize}
\item [(1)] it assigns an acceptance probability for each data point and selects observations according to the assigned probabilities;
\item [(2)] it fits a multi-class logistic model with the subsampled observations to obtain an estimate of the unknown model parameter.
\end{itemize}
 In the proposed framework, the acceptance probability for each data point can be obtained using an \emph{arbitrary} probability function. Unlike correction based methods \cite{king2001logistic,fithian2014local} that are specialized for certain models such as linear model, we propose a maximum likelihood estimate (MLE) that integrates the correction into the MLE formulation, and this approach allows us to deal with arbitrary sampling probability and produces a consistent estimator within the original model family as long as the underlying logistic model is correctly specified. This new integrated estimation method avoids the post-estimation correction step used in the existing literature.

Based on this estimation framework, we propose a new sampling scheme that generalizes LCC sampling, which we describe briefly in the following. Given a rough but consistent prediction $\tilde{p}(y|\x)$ of the true probability distribution $p(y|\x)$, this scheme preferentially chooses data points with labels that are conditionally uncertain given their local observations $\x$ based on $\tilde{p}(\cdot|\cdot)$. The proposed sampling strategy is therefore referred to as \emph{Local Uncertainty Sampling} (LUS). We show that the LUS estimator can achieve an asymptotic variance that is never worse than that of the uniform random sampling. That is, we can achieve variance of no more than $\gamma$ $(\gamma \geq 1)$ times the variance of the full-sample based MLE by using no more than $1/\gamma$ of the sampled data in expectation. Moreover the required sample size can be significantly smaller than $1/\gamma$ of the full data when the accuracy of the rough estimate $\tilde{p}(y|\x)$ is high. This generalizes the result for LCC in \cite{fithian2014local}, which reaches a similar conclusion for binary logistic regression when $\gamma \geq 2$.

We also study the case when the model is misspecified. In this case, for binary classification, LUS has the same properties as those of LCC that the subsampling based estimator is consistent to the best estimator for the original population given a consistent pilot estimate. Unfortunately, for general multi-class problems, the LUS estimator is biased. Nevertheless, we empirically find that the LUS method works well for both binary classification and multi-class classification problems even when the model is misspecified.
%\DIFdelbegin \DIFdel{However, we show that its bias and variance inflation are asymptotically negligible if the probability estimate based on the best estimator from the misspecified model is close to the true probability distribution.}\DIFdelend

We conduct extensive empirical evaluations on both simulated and real-world datasets, showing that the experimental results match the theoretical conclusions and the LUS method significantly outperforms the previous approaches.

Our main contributions can be summarized as follows:
\begin{itemize}
	\item we propose a general estimation framework for large-scale multi-class logistic regression, which can be used with
          arbitrary sampling probabilities. The procedure always generates a consistent estimator within the original model family when the model is correctly specified.
          This method can be applied to general logistic models without the need of post-estimation corrections;
	\item under this framework, we propose an efficient sampling scheme which we refer to as local uncertainty sampling.
          For any $\gamma \geq 1$, the method can achieve asymptotic variance no more than that of the uniform subsampling with sampling probability $1/\gamma$, using an expected sample size no more than that of the uniform subsampling.
          Moreover, the required sample size can be significantly smaller than that of the uniform subsampling when the classification accuracy of the underlying problem is high.
 \end{itemize}

\section{Preliminaries of Multi-Class Logistic Regression}
\label{sec:prelim}
For a $K$-class classification problem, we observe data points $(\x,y) \in \R^d \times\{1,2,\ldots, K\}$ from an unknown underlying distribution $\D$,  where $\x$ is the feature vector and $y$ is the corresponding label.
% The label $c$ can be alternatively represented by a $K$-dimensional vector $(y_1, y_2, \cdots, y_{K})^\top$ with only one non-zero element $y_c=1$ at the corresponding class label $c$.
Given a set of $n$ independently drawn observations $\{(\x_i,y_i): i=1,\ldots, n\}$ from $\D$, we want to estimate $K$ conditional probabilities $\mathbb{P}_\D(Y=k|\bm{X}=\x)$ for $k=1, 2, \ldots, K$. This paper considers multi-class logistic model with the following parametric form:
{\small
\begin{align*}
	&\P_\D(Y=k|\bm{X}=\x) = \frac{e^{f(\x,\bt_k)}}{1+\sum_{k'=1}^{K-1}e^{f(\x,\bt_{k'})}}\quad\mathrm{for}~ k=1,  \ldots, K-1, \\
	&\P_\D(Y=K|\bm{X}=\x) = \frac{1}{1+\sum_{k'=1}^{K-1}e^{f(\x,\bt_{k'})}} ,
\end{align*}
}\noindent
where $\bm{\bt}_k$ is the model parameter for the $k$-th class and $f$ is a known function.
The above model implies that 
{\small
\begin{align}
	f(\x,\bt_k)=\log{\frac{\P_\D(Y=k |\bm{X}=\x)}{\P_\D(Y=K|\bm{X}=\x)}}\quad \mathrm{for}~ k=1,\ldots,K-1 .
	\label{eq:log-odds}
\end{align}
}\noindent
Let $\bT=(\bt^{\top}_1, \bt_2^{\top}, \ldots, \bt_{K-1}^{\top})^{\top}$ be the parameter for the entire model. Eq.~(\ref{eq:log-odds}) is specified in terms of $K-1$ log-odds with the constraint that the probabilities of each class should sum to one. Note that the logistic model uses one reference class as the denominator in the odds-ratios, and the choice of the denominator is arbitrary since the estimates are equivalent under this choice. We use the $K$-th class as the reference class throughout the paper.

When the underlying model is correctly specified, there exists a true parameter $\bT^0=(\bt^{0\top}_1, \bt_2^{0\top},$ $\ldots,\bt_{K-1}^{0\top})^{\top}$ such that
{\small
\begin{align}
  f(\x,\bt_k^0)=\log\frac{\P_\D(Y=k|\bm{X}=\x)}{\P_\D(Y=K|\bm{X}=\x)}.
  \label{eq:correctly-specified-condition}
\end{align}
}\noindent
Moreover, $\bT^0$ is the maximizer of the expected population likelihood:
{\small
\begin{align}
 L(\bT)=
 \E_{\x,y\sim\D}\left[\sum_{k=1}^{K-1}\mathbb{I}(y=k) \cdot f(\x,\bt_k)-\log\left(1+\sum_{k=1}^{K-1}e^{f(\x,\bt_k)}\right)\right],
 \label{eq:risk}
\end{align}
}\noindent
where $\mathbb{I}(\cdot)$ is an indicator function. In the maximum likelihood formulation of multi-class logistic regression, the unknown parameter $\bT^0$ is estimated from the data by maximizing the empirical likelihood:
{\small
\begin{equation}
  \hat{L}_n(\bT)=\frac{1}{n}\sum_{i=1}^n\left[\sum_{k=1}^{K-1}\mathbb{I}(y_i=k)\cdot f(\x_i,\bt_k)-\log\left(1+\sum_{k=1}^{K-1}e^{f(\x_i,\bt_k)}\right)\right].
 \label{eq:logistic}
\end{equation}
}\noindent

For large-scale multi-class logistic regression problems, $n$ can be extremely large. In such cases, solving the multi-class logistic regression problem \eqref{eq:logistic} may be computationally infeasible due to limited computational resources. To overcome this computational challenge, we  propose a subsampling framework in the following section.

\section{Subsampling based Estimation}
\label{sec:method}
In this section, we introduce the estimation framework with subsampling for large-scale multi-class logistic regression. There are two main steps:
\begin{enumerate}
\renewcommand{\labelenumi}{(\theenumi)}
	\item for every data point $(\x,y)$, suppose that the arbitrary sampling probability function $a(\x,y) \in [0,1]$ is given.
          For each pair of observation $(\x_i,y_i)$ ($i=1,\ldots,n$), generate a random binary variable $z_i \in \{0,1\}$,  drawn from the $\{0,1\}$-valued Bernoulli distribution $\B(\x_i,y_i)$ with acceptance probability
{\small
\[
	\P_{\B(\x_i,y_i)}(z_i=1)=a(\x_i,y_i);
\]
}\noindent
	\item keep the samples with $z_i=1$ for $i \in \{1,\ldots,n\}$. Fit a multi-class logistic regression model based on the selected samples by solving the optimization problem
\end{enumerate}
{\small
\begin{align}
\underset{\bT}{\max} ~
\frac{1}{n}\sum_{i=1}^n z_i\left[\sum_{k=1}^{K-1}\mathbb{I}(y_i=k)
f(\x_i,\bt_k)
-\log\left(1+\sum_{k=1}^{K-1}\frac{a(\x_i,k)}{a(\x_i,K)}e^{f(\x_i,\bt_k)}\right)\right].
\label{eq:model}
\end{align}
}

We now derive the above procedure under the assumption that
the logistic model is correctly specified as in Eq. (\ref{eq:correctly-specified-condition}).
As we will show later, the acceptance probability used in the first step can be an arbitrary function, and the above method always produces a consistent estimator for the original population. The computational complexity in the second step is reduced to fitting the model with $\sum_{i=1}^n z_i$ samples after the subsampling step.

Given $(\x,y)\sim\D$, we draw a random variable $z$ according to the Bernoulli distribution $\B(\x,y)$. This gives the following augmented distribution $\A$ for the joint variable $(\x,y,z) \in \mathbb{R}^d \times \{1,2,\ldots,K\} \times \{0,1\}$ with probability
{\small
\begin{align*}
\P_{\A}(\bm{X}=\x, Y=k, Z=z)
=
\P_\D(\bm{X}=\x, Y=k) [a(\x, k) \mathbb{I}(z=1)+ (1-a(\x, k)) \mathbb{I}(z=0)].
\end{align*}
}\noindent
Note that each sampled data pair follows $(\x_i,y_i) \sim \D$, and the random variable $z_i$ is independently drawn from $\B(\x_i, y_i)$. It follows that each joint data point $(\x_i,y_i,z_i)$ is drawn i.i.d. from the distribution $\A$. For the sampled data $(\x_i,y_i)$ with $z_i=1$, the distribution of random variable $(\x, y)$ follows from
{\small
\[
\P_{\A}(\bm{X}=\x, Y=k|Z=1) \propto \P_\D(\bm{X}=\x, Y=k) a(\x, k).
\]
}\noindent
Therefore, we have
{\small
\begin{align*}
\log\frac{\P_{\A}(Y=k|\bm{X}=\x, Z=1)}{\P_{\A}(Y=K|\bm{X}=\x, Z=1)}
=
\log\frac{\P_\D(Y=k|\bm{X}=\x)}{\P_\D(Y=K|\bm{X}=\x)}+\log\frac{a(\x, k)}{a(\x, K)} .
\end{align*}
}\noindent
If $f$ is correctly specified for $\D$, then the following function family
{\small
\begin{align}
	g(\x,\bt_k)=f(\x,\bt_k)+\log\frac{a(\x,k)}{a(\x,K)}
\label{eq:logodds}
\end{align}
}\noindent
is correctly specified for $\A$, i.e., the true parameter $\bT^0$ in Eq. (\ref{eq:correctly-specified-condition})
also satisfies
{\small
\[
g(\x,\bt_k^0)=\log\frac{\P_{\A}(Y=k|\bm{X}=\x, Z=1)}{\P_{\A}(Y=K|\bm{X}=\x, Z=1)} \quad \mathrm{for}~k=1, \ldots, K-1.
\]
}\noindent
Therefore, we have the following logistic model under $\A$:
{\small
\begin{align*}
 \P_{\A}(Y=k|\bm{X}=\x, Z=1)
 &=\frac{e^{g(\x,\bt_k)}}{1+\sum_{k=1}^{K-1} e^{g(\x,\bt_k)}} \quad \mathrm{for}~k= 1,  \ldots, K-1, \\
 \P_{\A}(Y=K|\bm{X}=\x, Z=1)&=\frac{1}{1+\sum_{k=1}^{K-1} e^{g(\x,\bt_k)}}.
\end{align*}
}\noindent
It follows that given arbitrary sampling probability function $a(\x,y) \in [0,1]$, $\bT^0$ can be obtained by using MLE with respect to the new population $\A$:
{\small
\begin{align}
\underset{\bT}\max~R(\bT): 
=\E_{\x,y,z \sim \A}\ z\left[\sum_{k=1}^{K-1}\mathbb{I}(y=k)\cdot g(\x,\bt_k) - \log \left(1 + \sum_{k=1}^{K-1}e^{g(\x,\bt_k)}\right) \right].
\label{eq:LUSriskPop}
\end{align}
}\noindent
In practice, the model parameter $\bT^0$ can be estimated by empirical conditional MLE with respect to the sampled data $\{(\x_i,y_i,z_i): i=1,\ldots,n\}$:
{\small
\begin{align}
\underset{\bT}\max~\hat{R}_n(\bT):
=\frac{1}{n}\sum_{i=1}^n z_i\left[\sum_{k=1}^{K-1}\mathbb{I}(y_i=k)\cdot g(\x_i,\bt_k)-\log\left(1+\sum_{k=1}^{K-1}e^{g(\x_i,\bt_k)}\right)\right],
\label{eq:MLE_lus}
\end{align}
}\noindent
which is equivalent to Eq. (\ref{eq:model}). Let $\hat{\bT}_{Sub} =\arg\max_{\bT}\hat{R}_n(\bT)$ be the subsampling based estimator. As we will see in the next section, $\hat{\bT}_{Sub}$ is a consistent estimator of $\bT^0$ when the model is correctly specified.

\section{Asymptotic Analysis}\label{sec:theory}

In this section, we examine the asymptotic behavior of the proposed method described in Section~\ref{sec:method}. All of the proofs are provided in Appendix \ref{appendix:proof}.

\subsection{Consistency and Asymptotic Distribution}
First, based on the empirical likelihood in Eq. (\ref{eq:MLE_lus}), we have the following result for $\hat{\bT}_{Sub}$.

\begin{thm}[Consistency and Asymptotic Normality]
\label{thm:variance}
Suppose that the parameter space is compact with $\P_{\D}(f(\x,\bT) \neq f(\x,\bT^0)) >0$ for all $\bT\neq\bT^0$ and $\x\sim\D$. Moreover, assume the quantities $\|\nabla_{\bt_k} f(\x,\bt_k)\|$, $\|\nabla_{\bt_k}^2 f(\x,\bt_k)\|$ and $\|\nabla_{\bt_k}^3 f(\x,\bt_k)\|$ for $k=1,\ldots,K-1$ are bounded under some norm $\|\cdot\|$ for any $\bT$.\footnote{\small This assumption can be further weakened with some more complex arguments, e.g., by using the moment assumption on the derivatives.} 
Let $p(\x, k)=\mathbb{P}_{\D}(Y=k|\bm{X}=\x)$. If Eq. (\ref{eq:correctly-specified-condition}) is satisfied, i.e., the model is correctly specified, then given an arbitrary sampling probability function $a(\x, y)$, as $n\to\infty$,  the following holds:
\begin{enumerate}
\renewcommand{\labelenumi}{(\theenumi)}
  \item $\hat{\bT}_{Sub}$ converges to $\bT^0$;
  \item $\hat{\bT}_{Sub}$ asymptotically follows the normal distribution:
{\small
  \begin{align}
  	\sqrt{n}\left(\hat{\bT}_{Sub}-\bT^0\right)\xrightarrow{d}\mathcal{N}\left(\bm{0},
  	\left[
\E_{\x\sim\D}
	\bm{\nabla}
	\mathbf{S}
	\bm{\nabla}^{\top}
\right]^{-1}
\right),
  \label{eq:finalvariance}
  \end{align}
  }\noindent
\end{enumerate}
where {\small $\bm{\nabla}=diag\left(\left[\nabla_{\bt_1}f(\x,\bt_1^0),\nabla_{\bt_2}f(\x,\bt_2^0),\ldots, \nabla_{\bt_{K-1}}f(\x,\bt_{K-1}^0)\right]\right)$} is a block diagonal matrix, each block of which is $\nabla_{\bt_k}f(\x,\bt_k^0)$, and
{\small
\begin{align}
	\mathbf{S}=
	diag\left(
	\left[
	\begin{array}{c}
		a_1p_1\\
		a_2p_2\\
		\vdots \\
		a_{K-1}p_{K-1} \\
	\end{array}
	\right]
	\right)-
	\frac{1}{\sum_{k=1}^K a_kp_k}
	\left[
	\begin{array}{c}
		a_1p_1  \\
		a_2p_2  \\
		\vdots  \\
		a_{K-1}p_{K-1} \\
	\end{array}
	\right]
	\left[
	\begin{array}{c}
		a_1p_1  \\
		a_2p_2  \\
		\vdots  \\
		a_{K-1}p_{K-1} \\
	\end{array}
	\right]^\top,
	\label{eq:S}
\end{align}
}\noindent
with notations $a_k$ and $p_k$ indicating $a(\x,k)$ and $p(\x,k)$, respectively.
\end{thm}

Theorem \ref{thm:variance} shows that given an arbitrary sampling probability $a(\x,k)$, the proposed method in Section \ref{sec:method} generates a consistent estimator $\hat{\bT}_{Sub}$ \emph{without} post-estimation correction. This is different from existing methods such as the LCC in \cite{fithian2014local}, which employs post-estimation corrections.
One benefit of the proposed method is that without post-estimation correction, we can still produce a consistent estimator in the original model family, and our framework allows different sampling functions for different data points $(\x_i,y_i)$. For example, in time series analysis,  we may want to sample recent data points more aggressively than past data points. This can be handled naturally by our framework, but not by using the post-estimation correction based approach.
Another benefit is that the framework can be naturally applied with regularization, which can be regarded as a restriction on the parameter space for $\bT$.

From Theorem \ref{thm:variance}, the estimator $\hat{\bT}_{Sub}$ is asymptotically normal with variance 
$\left[
\E_{\x\sim\D}
	\bm{\nabla}
	\mathbf{S}
	\bm{\nabla}^{\top}
\right]^{-1}$.
Given a data point $(\x,y)$, although the sampling probability $a(\x,y)$ can be arbitrary, it is important to select a sampling probability such that the variance of the resulting estimator is as small as possible. In the following, we study a specific choice of $a(\x,y)$ that achieves this purpose.

\subsection{Sampling Strategy}
Recall from Eq.~\eqref{eq:S} the definition of $\mathbf{S}$.
Let $\mathbf{S}_{full}$ be the corresponding matrix $\mathbf{S}$ when we set $a(\x,k)=1$ for all $k$, i.e., we accept all data points in the dataset. Then
{\small
\begin{align}
	\mathbf{S}_{full}=
	diag\left(\left[
	\begin{array}{c}
		p(\x,1) \\
		p(\x,2) \\
		\vdots \\
		p(\x,K-1) \\
	\end{array}
	\right]\right)
	-
	\left[
	\begin{array}{c}
		p(\x,1)  \\
		p(\x,2)  \\
		\vdots  \\
		p(\x,K-1) \\
	\end{array}
	\right]
	\left[
	\begin{array}{c}
		p(\x,1)  \\
		p(\x,2)  \\
		\vdots  \\
		p(\x,K-1) \\
	\end{array}
	\right]^\top.
	\label{eq:S-full}
\end{align}
}\noindent
%%%%%%%%%%%%%%%%%%%%%%%%%%%%%%%%%%%%%%%%%%%%%%%%
If we set $a(\x,k)=\frac{1}{\gamma}$ for all $k$ for some $\gamma\geq1$, i.e., we sample a fraction of $\frac{1}{\gamma}$ of the full dataset uniformly at random, we denote the corresponding matrix as $\mathbf{S}_{US:\frac{1}{\gamma}}$. Then, $\mathbf{S}_{US:\frac{1}{\gamma}}=\frac{1}{\gamma}\mathbf{S}_{full}$.
In the following, we denote the asymptotic variances of $\hat{\bT}_{Sub}$, the full-sample based estimator and the estimator obtained from $\frac{1}{\gamma}$ uniformly sampled data as
{\small
\[
\mathcal{V}_{Sub}=\left[\E_{\x\sim\D}\bm{\nabla}\mathbf{S}\bm{\nabla}^\top\right]^{-1},
\mathcal{V}_{full}=\left[\E_{\x\sim\D}\bm{\nabla}\mathbf{S}_{full}\bm{\nabla}^\top\right]^{-1}
\ \text{and}\ \mathcal{V}_{US:\frac{1}{\gamma}}=\gamma\mathcal{V}_{full},
\]
}\noindent
respectively.
Our purpose is to find a better sampling strategy with lower variance than that of uniform sampling.
That is, we want to choose an acceptance probability function $a(\x,k)$ such that there exists some scalar $\gamma\geq1$ making
{\small
\[
\mathcal{V}_{Sub}\preceq\gamma\mathcal{V}_{full}=\mathcal{V}_{US:\frac{1}{\gamma}}
\]
}\noindent
under the constraint that
{\small
\[
\E_{\x, y\sim\D}\  a(\x, y) \leq \frac{1}{\gamma}.
\]
}\noindent
The constraint requires the expected subsample size to be no more than $n/\gamma$, i.e., we sample no more than $1/\gamma$ fraction of the full data.
%%%%%%%%%%%%%%%%%%%%%%%%%%%%%%%%%%%%%%%%%%%%%%%%%%%%%

\begin{thm}[Sampling Strategy]\label{thm:LUS}
Use the same assumptions and definitions in Theorem \ref{thm:variance}. For any data point $\x$, let
	{\small
	\begin{align}
	q(\x)=\max\big(0.5, p(\x,1), \ldots, p(\x,K)\big).
	\label{eq:q}
	\end{align}
	}\noindent
        Given any $\gamma\geq1$, consider the following choice of the acceptance probability function:
	\begin{enumerate}
	\renewcommand{\labelenumi}{(\theenumi)}
		\item for $\gamma\geq2q(\x)$,  set $a(\x,k)$ as
{\small
	\begin{align}
		a(\x,k)=
		\left\{
		\begin{array}{ll}
			\frac{2\left(1-q(\x)\right)}{\gamma}, & \text{if } p(\x,k)=q(\x)\geq 0.5 \\
			\frac{2q(\x)}{\gamma}, & \text{otherwise} \\
		\end{array}
		\right.,
		\ k=1,\ldots,K;
		\label{eq:a1}
	\end{align}
}\noindent
	\item for $1\leq\gamma<2q(\x)$, set $a(\x,k)$ as
{\small
	\begin{align}
		a(\x,k)=
		\left\{
		\begin{array}{ll}
			\frac{1-q(\x)}{\gamma-q(\x)}, & \text{if } p(\x,k)=q(\x)\geq 0.5 \\
			1, & \text{otherwise} \\
		\end{array}
		\right.,
		\ k=1,\ldots,K.
		\label{eq:a2}
	\end{align}
}
	\end{enumerate}
	Then, we always have
{\small
	\begin{align}
			\mathcal{V}_{Sub}\preceq\gamma\mathcal{V}_{full}=\mathcal{V}_{US:\frac{1}{\gamma}},
			\label{eq:relation}
	\end{align}
}\noindent
and the expected number of subsampled examples is
{\small
	\begin{align}
	n_{Sub}=n\E_{\x, y\sim\D}\ a(\x, y)
	\leq\frac{n}{\gamma}.
	\label{eq:expected}
	\end{align}
}
\end{thm}

It is easy to check that the assigned acceptance probability in Theorem \ref{thm:LUS} is always valid, i.e., it is a value in $[0,1]$.
With the sampling strategy in Theorem~\ref{thm:LUS}, we always use no more than a fraction $1/\gamma$ of the full data to obtain an estimator with variance no more than $\gamma$ times the variance of the full-sample based MLE. In other words, the proposed method is never worse than
the uniform sampling method. Moreover, the required sample size
$n_{Sub}$ can be significantly smaller than $n/\gamma$ under favorable conditions.
%For example, when $\gamma \geq 2$, it is easy to verify that
%$n_{Sub}/ (n/\gamma) = \E_{\x\sim\D}\ 4 q(\x) (1-q(\x))$, which is close to zero if $q(\x) \approx 1$ for most of the $\x$. This happens when the classification accuracy is high.

More precisely, we have the following formula for the expected conditional sampling probability:
{\small
\[
\E_{(\x, y) \sim\D}\ a(\x, y)=
\Bigg\{
  \begin{array}{ll}
    \frac{4}{\gamma} \ q(\x)(1-q(\x)), & \text{under case (1) in Theorem \ref{thm:LUS}}, \\
    \ \frac{\gamma(1-q(\x))}{\gamma-q(\x)}, & \text{under case (2) in Theorem \ref{thm:LUS}}. \\
  \end{array}
\]
}\noindent
Therefore, in favorable case where most $q(\x)\approx1$ for $\x\sim\D$, i.e., the data are conditionally imbalanced, our method will subsample very few examples to achieve the desired variance compared to that of $1/\gamma$ random sampling.

An intuitive explanation for this sampling strategy is that if there exists a class $k$ that dominates the other classes for any given $\x$, i.e., $p(\x,k)\geq0.5$, then the sampling probability will be proportional to `$1-accuracy$'. That is, when the classification accuracy for the problem is high, the proposed sampling method will significantly outperform the random sampling.

For binary classification problem ($K=2$), Theorem~\ref{thm:LUS} reduces to LCC sampling in \cite{fithian2014local} when $\gamma \geq 2$. Although a method to achieve a desired variance for the case of $\gamma \in [1,2)$ was also proposed in \cite{fithian2014local}, it is different from our sampling strategy. Moreover, it does not prove that the number of samples needed by LCC for achieving such a desired variance is never worse than that of the uniform sampling for the case of $\gamma\in[1,2)$.
In fact, the empirical performance of LCC can be worse than our method under the case of $\gamma \in [1,2)$ as we will show in the experimental section.

In the multi-class case, our method is not a natural extension of the LCC sampling, which would imply a method to set all class probabilities to $1/K$ after sampling.
Instead, we will only assign a smaller sampling probability for $(\x,y)$
when $p(\x,y) \geq 0.5$. The method is less likely to select a sample when the label $y$ coincides with the the prediction of the underlying true model, while it will likely be selected if $y$ contradicts the underlying true model.
Since the sampling strategy prefers data points with uncertain labels,  we refer it to as {\em Local Uncertainty Sampling} (LUS). In the following, we will indicate the estimator as $\hat{\bT}_{LUS}$ if the acceptance probability is set according to Eqs.~(\ref{eq:a1}) and (\ref{eq:a2}) (recall that we use the notation $\hat{\bT}_{Sub}$ with an arbitrary sampling function $a(\x,y)$).

%%%%%% Randomness on the acceptance probability
\subsection{Randomness on the Acceptance Probability}
In Theorem \ref{thm:LUS}, $a(\x,k)$ is set as a function of the true probability $p(\x,k)$ for $k=1,\ldots,K$. 
In practice, we never know $p(\x,k)$, because if we do there would not be a need to estimate the model parameters. 
Instead, we can use a roughly estimated probability $\tilde{p}(\x,k)$ to compute the acceptance probability, which we refer to as the \emph{pilot} estimate.

There are multiple ways to estimate $\tilde{p}(\x,k)$. For example, one can obtain a pilot estimator $\bm{\lambda}$ by fitting the model using a smaller independent data set, or obtain $\tilde{p}(\x,k)$ from a different and simpler parametric model or non-parametric methods.
This is different from the LCC method, which involves post-estimation correction that  relies on an explicit pilot estimator $\bm{\lambda}$ that is additive to the original model parameter. 
For simplicity, in the following analysis, we still assume that $\tilde{p}(\x,k)$ is computed from a pilot estimator $\bl$.

%To emphasize the dependencies , we rewrite $\mathbf{S}$ and $\mathbf{\Delta}$ as $\mathbf{S}(\bm{\lambda},\mathbf{\Theta})$ and $\mathbf{\Delta}(\boldsymbol{\Theta})$   

Recall from Theorem~\ref{thm:variance} that given an acceptance probability, the asymptotic variance for $\hat{\boldsymbol{\Theta}}_{LUS}$ is dependent on $\mathbf{S}$ and $\boldsymbol{\nabla}$.
Since $a(\x,k)$ is computed based on $\bl$, from Eq.~\eqref{eq:S}, we see that $\mathbf{S}$ is now a function of $\bm{\lambda}$ and the model parameter $\bT$.  We rewrite $\mathbf{S}$ as $\mathbf{S}(\bm{\lambda},\bT)$ to emphasize its dependency on $\bm{\lambda}$ and $\bm{\Theta}$. 
Similarly, we represent $\bm{\nabla}$ as $\bm{\nabla}(\bT)$. 
Moreover, with the use of the pilot estimator, from Eq.~\eqref{eq:LUSriskPop}, $R(\cdot)$ is now a function of $\bm{\lambda}$ and $\bT$, and we rewrite it as $R(\bm{\lambda},\bT)$.

In the following, we characterize the asymptotic distribution of $\hat{\boldsymbol{\Theta}}_{LUS}$ when the pilot estimator $\boldsymbol{\lambda}$ is estimated from an independent data set.
%Thus, if $\bm{\lambda} = \bT^0$, we have $\mathcal{V}_{LUS}=[\E_{\x\sim\D}\bm{\nabla}(\bT^0)\bm{S}(\bT^0, \bT^0)\bm{\nabla}(\bT^0)^\top]^{-1}$. 

\begin{cor}\label{cor:pilot}
Let $\tilde{p}(\x,k)$ be a probability estimate computed using a pilot estimator $\hat{\bl}$. If 
$\hat{\bl}\stackrel{p}{\rightarrow}\bT^0$ such that $\tilde{p}(\x,k)\stackrel{p}{\rightarrow} p(\x,k)$,
%$\sqrt{n}(\bl-\bT^0)\xrightarrow{d}\mathcal{N}(\bm{0},\mathcal{V}_{\lambda})$ for some variance $\mathcal{V}_{\lambda}$, 
we have
{\small
\[
\sqrt{n}\left(\hat{\bT}_{LUS}-\bT^0\right)\xrightarrow{d}\mathcal{N}\left(\bm{0}, \mathcal{V}_{LUS}\right),
\]
where $\mathcal{V}_{LUS}=[\E_{\x\sim\D}\bm{\nabla}(\bT^0)\mathbf{S}(\bT^0, \bT^0)\bm{\nabla}(\bT^0)^\top]^{-1}$ is a constant that is independent of the pilot estimator $\hat{\bl}$.
}
\end{cor}
Corollary~\ref{cor:pilot} implies that as long as the pilot estimator $\bl$ is a consistent estimator of $\boldsymbol{\Theta}^0$ such that $\tilde{p}(\x,k)$ converges to $p(\x,k)$, the randomness induced by the pilot estimator $\bl$ will not inflate the asymptotic variance of $\hat{\boldsymbol{\Theta}}_{LUS}$.  
The result is in conjunction with that of LCC when the model is correctly specified, and is also a generalization of their results to the case of $K>2$.

\subsection{Model Misspecification}
In practice, the model in Eq. (\ref{eq:correctly-specified-condition}) may not be correctly specified, so a true parameter $\bT^0$ may not exist. 
Under such case, we denote the best estimator obtained by maximizing Eq. (\ref{eq:risk}) for the original population $\D$ as $\bT^*$ and denote the corresponding probability estimate as $p^*(\x,k)$ for $k=1,\ldots,K$. Since the model is misspecified, we know that $p^*(\x,k)$ may not equal to $p(\x,k)$.

In the following, we study the properties of LUS and consider the cases when $K=2$ and $K>2$ separately. 
To distinguish between the two cases, we use lower-case letter $\bt$ to denote the model parameter for $K=2$,  i.e., binary classification problems. We first have the following results for $K=2$.

\begin{prop}\label{prop:mis1}
For $K=2$, suppose that the parameter space of $\bt$ is compact that $\P_{\D}(f(\x,\bt) \neq f(\x,\bt^*)) >0$ for all $\bt\neq\bt^*$ and $\x\sim\D$. Under model misspecification, let $\bl = \bt^*$ and we have
{\small
\[
\bt^*=\underset{\bt}{\arg\max} ~R(\bt^*,\bt).
\]
}\noindent
Moreover, $\bt^*$ is the unique maximizer of $R(\bt^*,\bt)$ and the LUS estimator is a consistent estimator of $\bt^*$.
\end{prop}
Proposition~\ref{prop:mis1} implies that under model misspecification, if we use a perfect pilot estimate $\bl = \bt^*$ to compute the acceptance probability, then $\hat{\bt}_{LUS}$ would also converge to $\bt^*$. Again, in practice, we never know $\bt^*$ a priori. Therefore, we will need a pilot estimator $\hat{\bl}$.   In the following proposition, we study how randomness of $\hat{\bl}$ affects the performance of the LUS estimator when the model is misspecified.

\begin{prop}\label{prop:mis2}
Assume that the pilot estimator $\hat{\bl}$ is such that $\sqrt{n}(\hat{\bl}-\bt^*)\xrightarrow{d}\mathcal{N}(\bm{0},\mathcal{V}_{\lambda})$ for some $\mathcal{V}_{\lambda}>0$.  Set the acceptance probability acceptance probability according to Eqs.~(\ref{eq:a1}) and (\ref{eq:a2}).
Under the same conditions in Proposition~\ref{prop:mis1}, when $K=2$ and under model misspefication, we have
{\small
\[
\sqrt{n}(\hat{\bt}_{LUS}-\bt^*)\xrightarrow{d}
\mathcal{N}\left(\bm{0}, \H(\bt^*,\bt^*)^{-1}[J(\bt^*,\bt^*)+C(\bt^*,\bt^*)\mathcal{V}_{\lambda}C(\bt^*,\bt^*)^{\top}]\H(\bt^*,\bt^*)^{-1}\right),
%\mathcal{N}\left(\bm{0}, \H^{-1}(J+CVC^{\top})\H^{-1}\right),
\]
}
where
{\small
\[
\H(\bl,\bt)=-\nabla_{\bt}^2R(\bl,\bt), \ 
J(\bl,\bt)=Var(\sqrt{n}\nabla_{\bt}R(\bl,\bt)), \ 
C(\bl,\bt)=\frac{\partial^2R(\bl,\bt)}{\partial\bt\partial\bl}.
\]
}
\end{prop}
\noindent
Proposition \ref{prop:mis2} implies that as long as the pilot estimator is consistent to $\bt^*$, $\hat{\bt}_{LUS}$ is consistent to $\bt^*$. 
Moreover, under model misspecification, we see that there is an inflation in the asymptotic variance induced by the random pilot estimator $\hat{\bl}$. 
From Propositions \ref{prop:mis1} and \ref{prop:mis2}, when the model is misspecified for $K=2$, the LUS method has the same properties as those of the LCC method proposed in \cite{fithian2014local}.  

%Unfortunately, Propositions \ref{prop:mis1} and \ref{prop:mis2} can not be generalized to the case of $K>2$. 
However, when $K>2$, the LUS estimator is biased.  
%In particular,  Proposition~\ref{prop:mis3} shows that $\bT^*$ is not the maximizer of $R(\bT^*,\bT)$. 
\begin{prop}\label{prop:mis3}
Under model misspecification, when $K>2$, the LUS estimator is biased even when $\bl=\bT^*$.  However, if there exists a class $k$ such that $p(\x,\bt_k^*)\geq0.5$, then $\bt_k^* 
={\arg\max}_{\bt_k}~ R(\bT^*,\bT_{-k}^*)$ where $\bT_{-k}^*$ is the parameter by setting $\bt_i=\bt_i^*$ for $i\neq k$. 
\end{prop}
Proposition~\ref{prop:mis3} implies that when $\bl=\bT^*$, if there exists a majority class such that $p(\x,\bt^*_k)\ge 0.5$, then the LUS estimator with respect to the $k$-th class will be consistent to $\bt_k^*$ when the other parameters $\bt_i$ are fixed as $\bt_i^*$ for $i\neq k$.
   Explicit characterizations of the bias for the LUS estimator when $K>2$ are complex and we put them in Appendix~\ref{appendix:bias}.

Nevertheless, in our empirical studies, we find that the LUS method works well for both binary classification and multi-class classification problems, even under model misspecification.

%%%%%%%%%%%%%%%%%%%%%%%%%%%%%%%%
% Local Uncertainty Sampling Algorithm
%%%%%%%%%%%%%%%%%%%%%%%%%%%%%%%%
\section{Local Uncertainty Sampling Algorithm}\label{sec:algo}

In order to apply the sampling strategy in Theorem \ref{thm:LUS} empirically, according to Corollary \ref{cor:pilot}, the main idea is to employ a rough but consistent estimate $\tilde{p}(\x,k)$ of the probabilities $p(\x,k)$ given $\x$, and then assign the acceptance probability according to $\tilde{p}(\x,k)$. In fact, we only need to have an approximate estimate of $q(\x)$ according to Eq. \ref{eq:q}; that is, we only need to roughly estimate the probability of the majority class given $\x$. 
Note that $\tilde{p}(\x,k)$ can be obtained via a consistent pilot estimator $\bm{\lambda}$ (which may be obtained using a small amount of uniformly sampled data) where $\tilde{p}(\x,k)= e^{f(\x,\bm{\lambda}_k)}/(1+\sum_{k'}^{K-1} e^{f(\x,\bm{\lambda}_{k'})})$. This is similar to what is employed in \cite{fithian2014local}. 
However, in practice, one may also use a simpler model family to obtain the pilot estimate, as shown in our MNIST experiment below, or one may use other techniques such as neighborhood based methods \cite{atkeson1997locally} for this purpose. 

In real-world applications, this rough estimate is often easy to obtain. For example, when data arrives in time sequence, a pilot estimator trained on previous observations can be used for fitting a new model when new observations are arriving. Moreover, a rough estimate obtained on a small subset of the full population can be used for training on the entire dataset. In our experiments, we adopt the latter: a small uniformly subsampled subset of the original population is used to obtain the rough estimate $\tilde{p}(\x,k)$.
As we will see later, this choice is sufficient for our method to obtain good practical performance. The LUS algorithm can be described in Algorithm~\ref{algo:LCCMC}.

\begin{algorithm}[t]
\caption{The LUS Algorithm for Multi-Class Logistic Regression.}
\label{algo:LCCMC}
\begin{algorithmic}[1]
\STATE Choose a desired $\gamma \geq 1$.
\STATE Given a rough estimate $\tilde{\bm{p}}_i=(\tilde{p}_{i,1},\ldots, \tilde{p}_{i,K})^\top$ for each data point $\x_i$, where $\tilde{p}_{i,k}$ is a roughly estimated probability of $\x_i$ belonging to class $k$.
\STATE Scan the data once and generate the random variables $z_i\sim\text{Bernoulli}(a(\x_i, y_i):z_i=1)$ based on the acceptance probability $a(\x_i, y_i)$ defined as
\[
a(\x_i,y_i)=
\left\{
\begin{array}{ll}
			\frac{1-\tilde{q}_i}{\gamma-\max(\tilde{q}_i,0.5\gamma)}, & \text{if } \tilde{p}_{i, y_i}=\tilde{q}_i\geq 0.5 \\
			\min(1,2\tilde{q}_i/\gamma), & \text{otherwise} \\
		\end{array}
		\right. ,
\]
where
$\tilde{q}_i=\max\big(0.5, \tilde{p}_{i,1},\ldots, \tilde{p}_{i,K}\big)$.
\STATE Fit a multi-class logistic regression model to the subsample set $\{(\x_i,y_i):z_i=1\}$ with the model function $g$ defined in Eq. (\ref{eq:logodds}):
\begin{align}
	\hat{\bT}_{LUS}=\arg\max_{\bT}\sum_{i=1}^n z_i\left(\sum_{k=1}^{K-1}\mathbb{I}(y_i=k)\cdot g(\x_i,\bt_k)-\log\left(1+\sum_{k=1}^{K-1}e^{g(\x_i,\bt_k)}\right)\right).
\end{align}
\STATE Output $\hat{\bT}_{LUS}$.
\end{algorithmic}
\end{algorithm}

\section{Experiments}
In this section, we evaluate the performance of the LUS method and compare it with the uniform sampling (US) and case-control (CC) sampling methods on both simulated and real-world datasets. For the CC sampling method, we extend the standard CC sampling considered in the binary classification problem to multi-class case by sampling equal number of data points for each class. Under marginal imbalance, if some minority classes do not have enough samples, we keep all data for those classes and subsample equal number of the remaining data points from other classes.
In addition, we also compare the LUS and LCC methods on the Web Spam dataset, which is a binary classification problem studied in \cite{fithian2014local}. The experiments are implemented on a single machine with 2.2GHz quad-core Intel Core i7 and 16GB memory.

\subsection{Simulation: Marginal Imbalance}\label{sec:simulation1}
We first simulate the case when the data is marginally imbalanced. We generate a 3-class Gaussian model according to $(\bm{X}|Y=k)\sim\mathcal{N}(\bm{\mu}_{k},\bm{\Sigma}_k)$, which is the true data distribution $\D$. We set the number of features as $d=20$, $\bm{\mu}_1=[\underbrace{1,1,\ldots,1}_{10}, \underbrace{0,0,\ldots,0}_{10}]^\top$, $\bm{\mu}_2=[\underbrace{0,0,\ldots,0}_{10}, \underbrace{1,1,\ldots,1}_{10}]^\top$, and $\bm{\mu}_3=[\underbrace{0,0,\ldots,0}_{20}]^\top$. The covariance matrices for classes $k=1,2,3$ are assigned to be the same, i.e.,  $\bm{\Sigma}_1=\bm{\Sigma}_2=\bm{\Sigma}_3=\mathbf{I}_d$, where $\mathbf{I}_d$ is a $d\times d$ identity matrix. 
So the true log-odds function $f$ is linear and 
we use a linear model to fit the simulated data, 
i.e., the model is correctly specified. 
Moreover, we set $\P(Y=1)=0.1$, $\P(Y=2)=0.8$, $\P(Y=3)=0.1$, which implies that the data is marginally imbalanced and the second class dominates the population.

Since the true data distribution $\D$ is known in this case, we directly generate the full dataset from the distribution $\D$. For the full dataset, we generate $n=50,000$ data points. The entire procedure is repeated for 200 times to obtain the variance of different estimators. For the LUS method, we randomly generate $n_{pilot}=5000$ data points (i.e., an amount of 10\% of the full data) from $\D$ to obtain a pilot estimator in \emph{every} repetition. Moreover, we generate another $n_{test}=100,000$ data points to test the prediction accuracy of different methods.

Recall that $\gamma$ controls the desired variance of the LUS estimator according to Theorem \ref{thm:LUS}. In the following experiments, we will test different values of $\gamma=\{1.1,1.2,\ldots,1.9,2,3,4,5\}$, respectively. Given the value of $\gamma$, suppose the LUS method will subsample a number of $n_{Sub}$ data points. Then, we let the US and CC sampling methods select an amount of $n_{Sub}+n_{pilot}$ examples to achieve fair comparison, because the LUS method has to pay for its usage of a random pilot estimate.

\begin{figure*}[!ht]
\centering
\subfigure[$\gamma=1.1$]{
\includegraphics[scale=0.1]{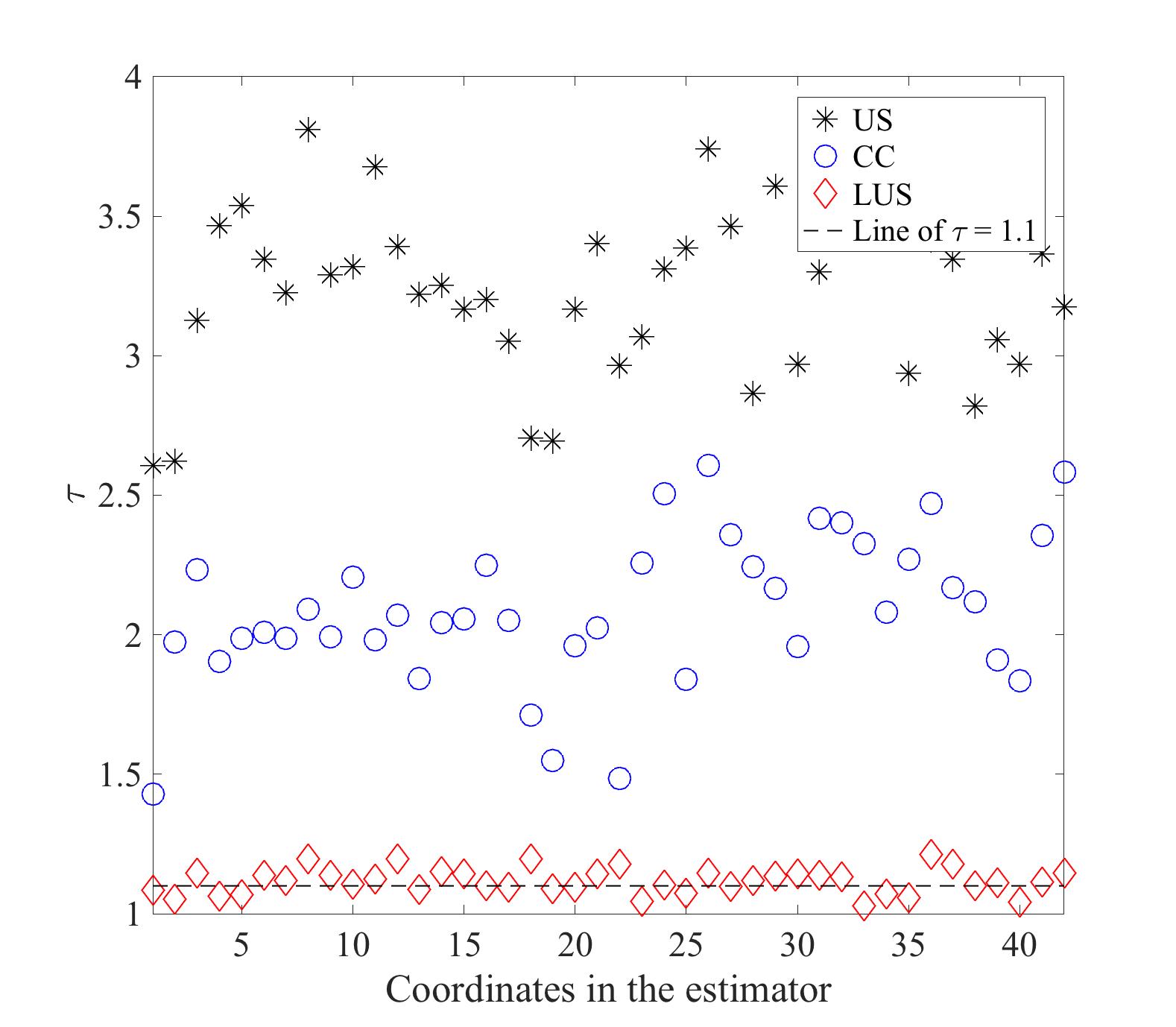}
}\hskip -0.2in
\subfigure[$\gamma=2$]{
\includegraphics[scale=0.1]{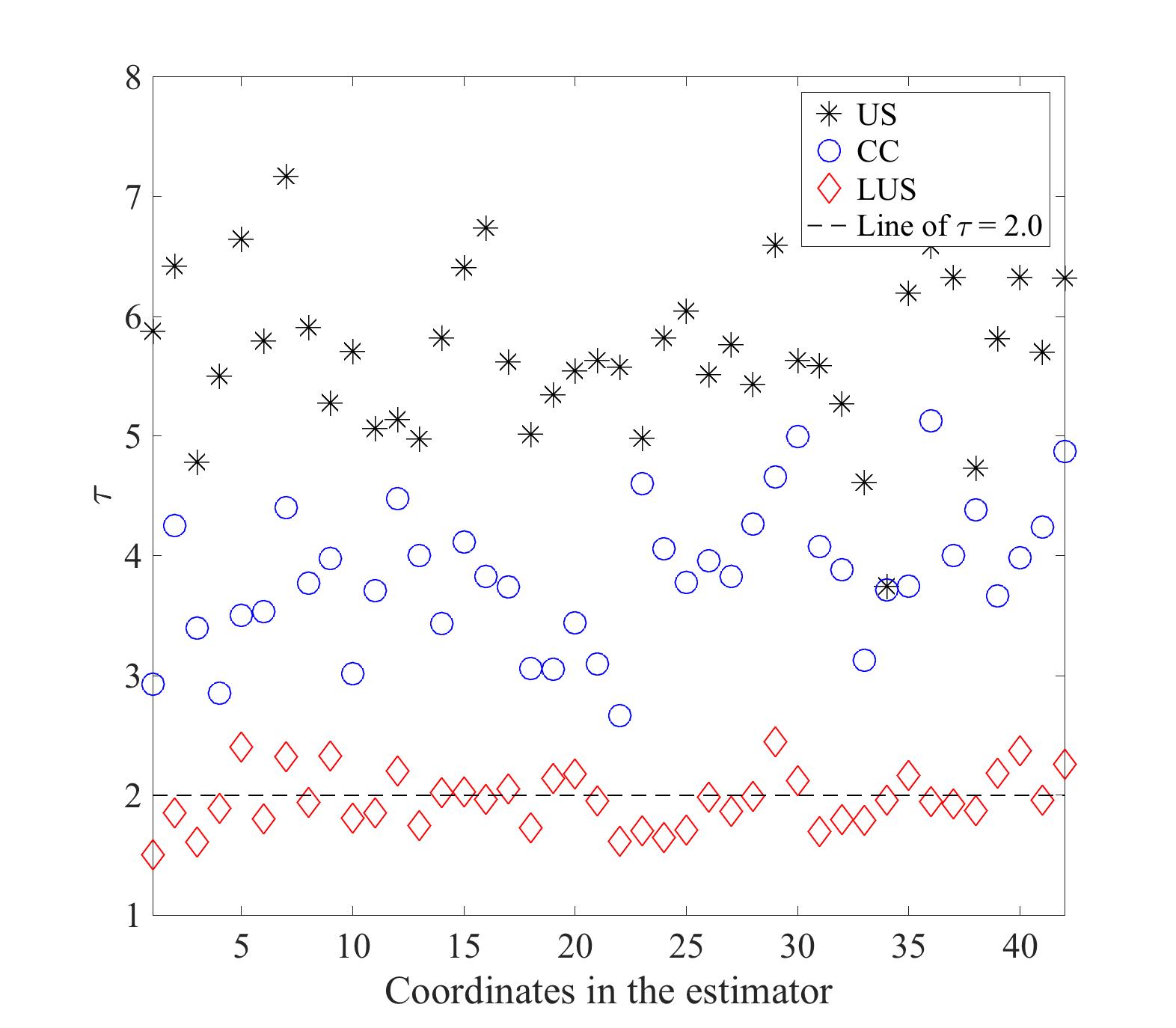}
}\hskip -0.2in
\subfigure[$\gamma=3$]{
\includegraphics[scale=0.1]{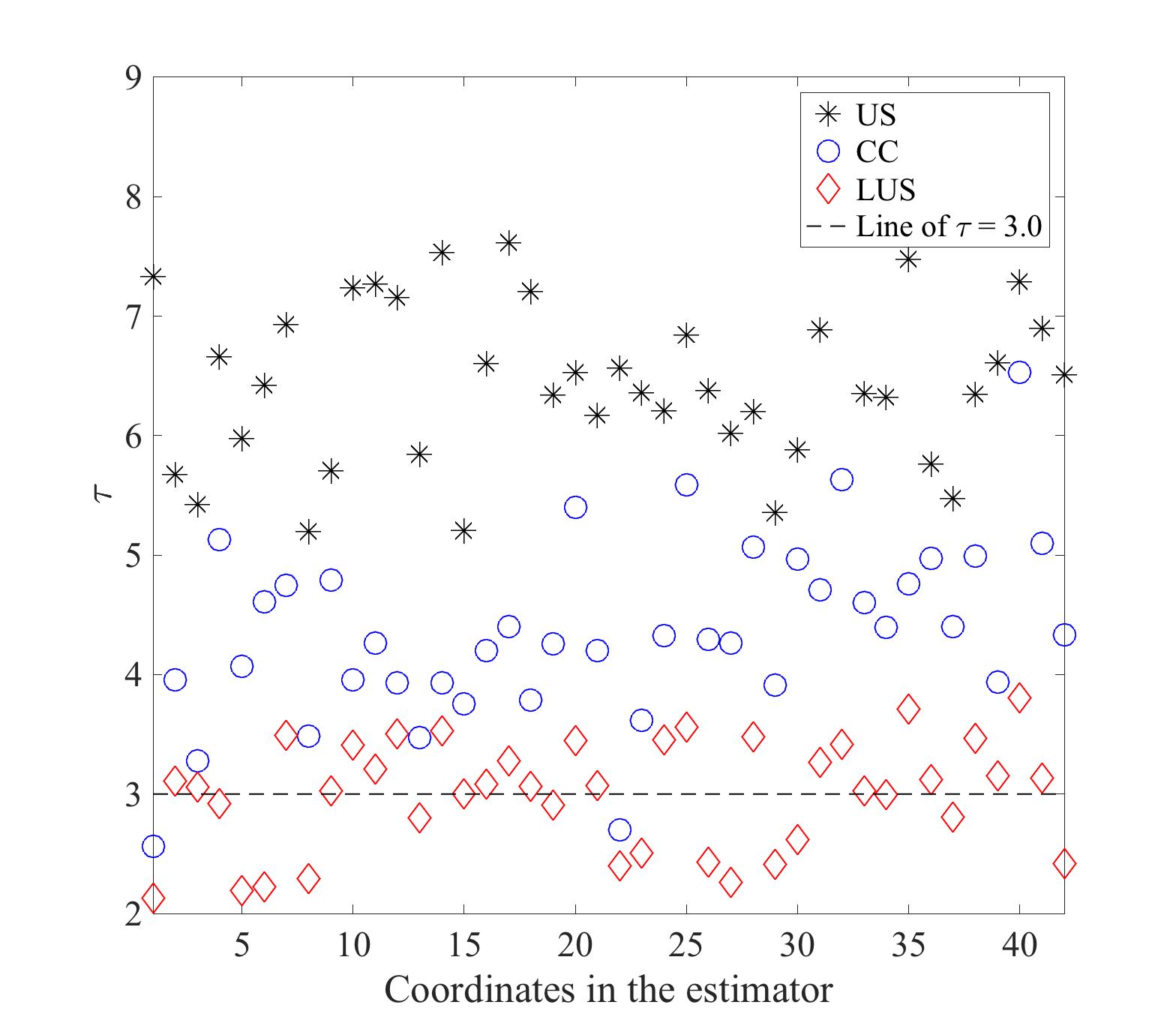}
}\vskip -0.1in
\caption{The $\tau$ value for each coordinate under different values of $\gamma$. $\tau$ denotes the ratio between the variance of each coordinate in the subsampling based estimator and the variance of the coordinate in the full-sample based MLE, i.e., $\tau=\text{Var}(\hat{\theta}_{Sub}) / \text{Var}(\hat{\theta}_{full})$.}
\label{fig:toY_imbalance1}
\end{figure*}

\begin{figure*}[!ht]
\centering
\subfigure[Average $\tau$ v.s. $\gamma$]{
\includegraphics[scale=0.38]{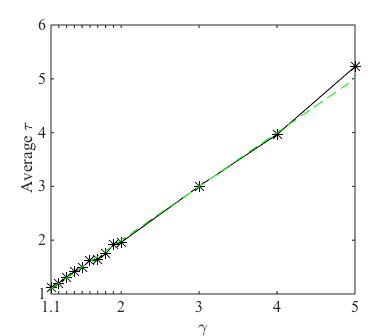}
}\hskip -0.15in
\subfigure[$n_{Sub}/n$ v.s. $\gamma$]{
\includegraphics[scale=0.38]{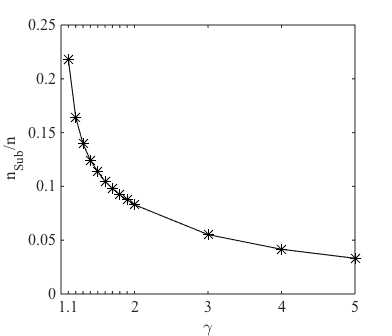}
}\hskip -0.1in
\subfigure[Accuracy v.s. $\frac{n_{Sub}+n_{pilot}}{n}$]{
\includegraphics[scale=0.38]{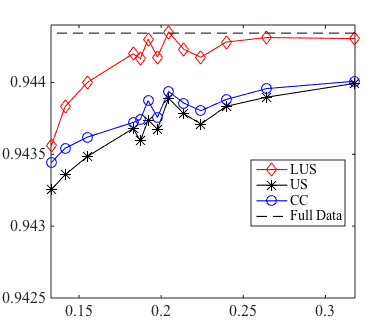}
}\vskip -0.1in
\caption{Plots in the first simulation.}
\label{fig:toY_imbalance2}
\end{figure*}

Since $\bt_k\in\R^{d}$ $(k=1,\ldots,K-1)$ in this case and there is an additional intercept parameter, the estimator contains a total number of $(d+1)(K-1)$ coordinates. Denote $\tau$ as the coordinate-wise ratio between the variance of the coordinate in the subsampling based estimator and the variance of the coordinate in the full-sample based MLE. We show the $\tau$ value for each coordinate under different values of $\gamma=\{1.1, 2,3\}$ in Fig. \ref{fig:toY_imbalance1}.
 In this simulation, there are $42$ coordinates. From the figures, we observe that the $\tau$ value for each coordinate of the LUS method is  around $\gamma$, which matches our theoretical analysis in Theorem \ref{thm:LUS}. On the other hand, the variances of the US and CC sampling methods are much higher than that of the LUS method, even when US and CC sample $n_{pilot}$ more data points than the LUS method.

In Fig. \ref{fig:toY_imbalance2}(a), we plot the relationship between the average $\tau$ for all coordinates against $\gamma$.  From the figure, we observe that the relationship is close to the $y=x$ line (the dashed green line), which implies that $\tau$ approximately equals $\gamma$. These experimental results well match our theoretical analysis. Fig. \ref{fig:toY_imbalance2}(b) reports the relationship between $n_{Sub}/n$ and $\gamma$. Fig. \ref{fig:toY_imbalance2}(c) shows the relationship between the prediction accuracy on the test data and the proportion of used training data $(n_{Sub}+n_{pilot})/n$. From the figure, when $(n_{Sub}+n_{pilot})/n$ decreases, the prediction accuracy of all the methods decreases, while the LUS method shows much slower degradation compared to the US and CC methods. Moreover, according to Fig. \ref{fig:toY_imbalance2}(c), we only need about $20\%$ of the full data (including those used for computing the pilot estimator) to achieve the same prediction accuracy as the full-sample based MLE, implying that the LUS method is very effective for reducing the computational cost while preserving high accuracy.

\subsection{Simulation: Marginal Balance}\label{sec:simulation2}

In this section, we generate marginally balanced data with conditional imbalance. Under this situation, the CC sampling method is identical to US, and hence we omit the CC sampling method from our comparison. The settings are exactly the same as those in the previous simulation, except that we let $\P(Y=1)=\P(Y=2)=\P(Y=3)=\frac{1}{3}$, which implies that the data is marginally balanced. 
However, as we will see later, this simulated data is conditionally imbalanced.

\begin{figure*}[!ht]
\centering
\subfigure[$\gamma=1.1$]{
\includegraphics[scale=0.1]{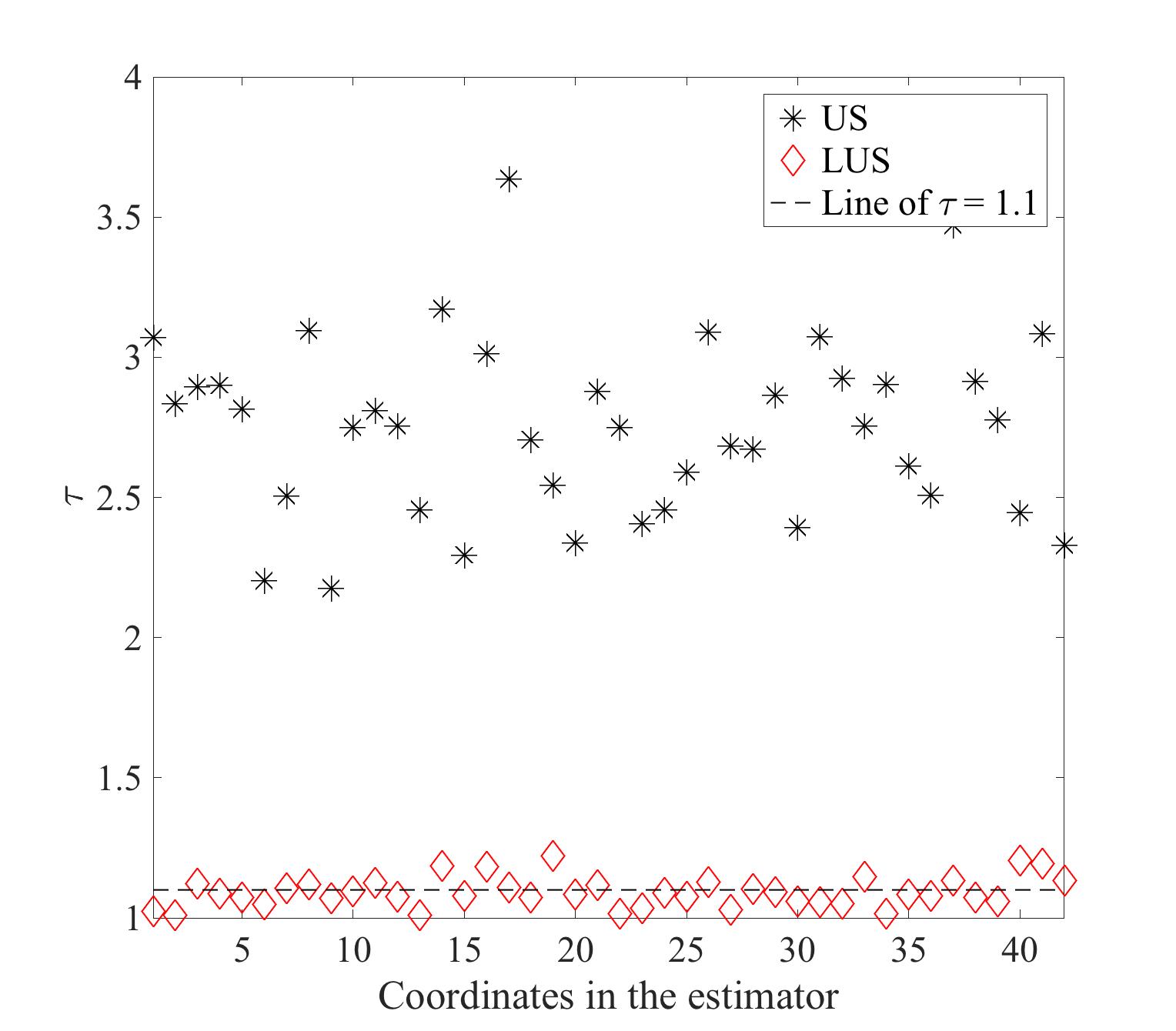}
}\hskip -0.2in
\subfigure[$\gamma=2$]{
\includegraphics[scale=0.1]{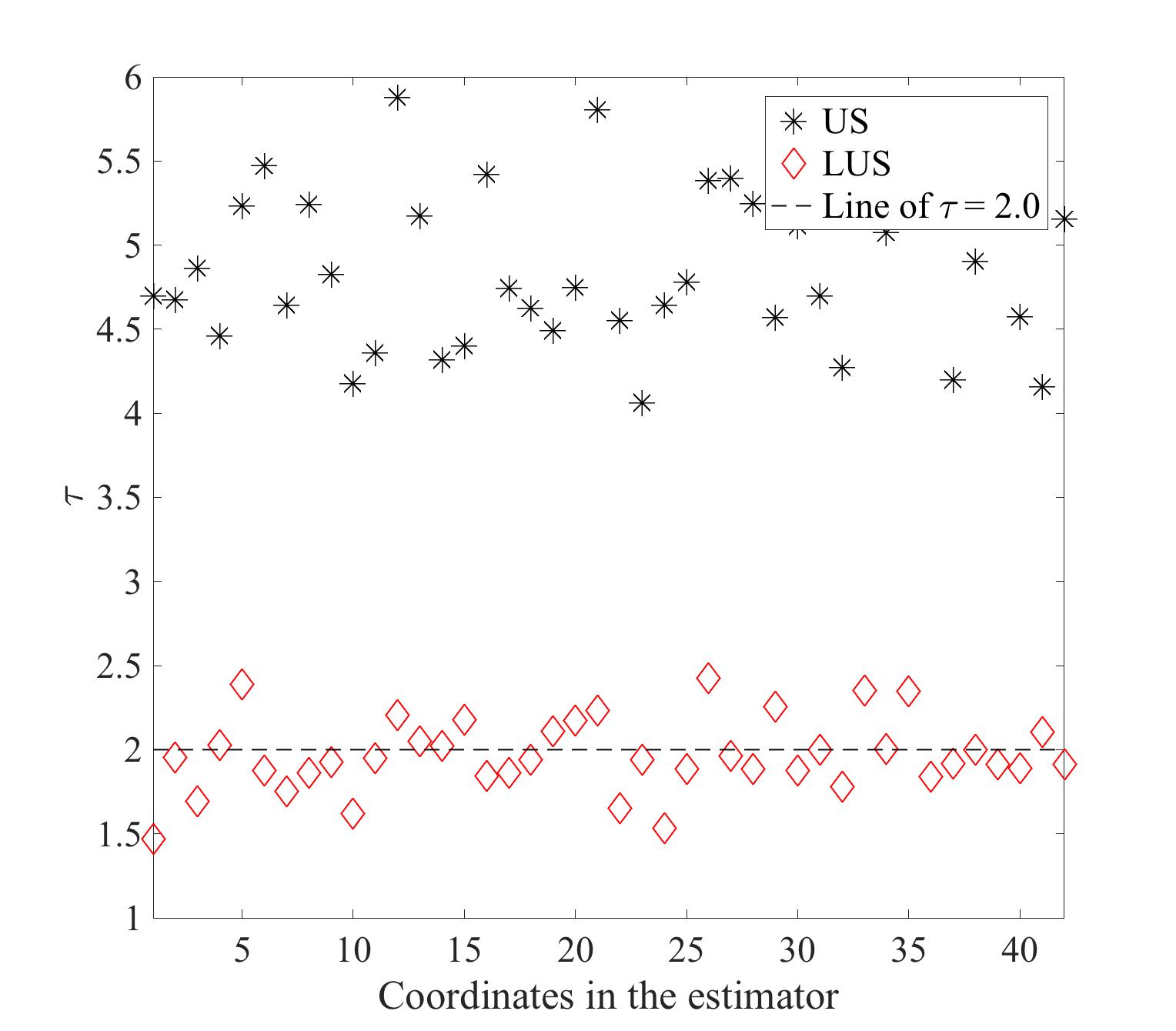}
}\hskip -0.2in
\subfigure[$\gamma=3$]{
\includegraphics[scale=0.1]{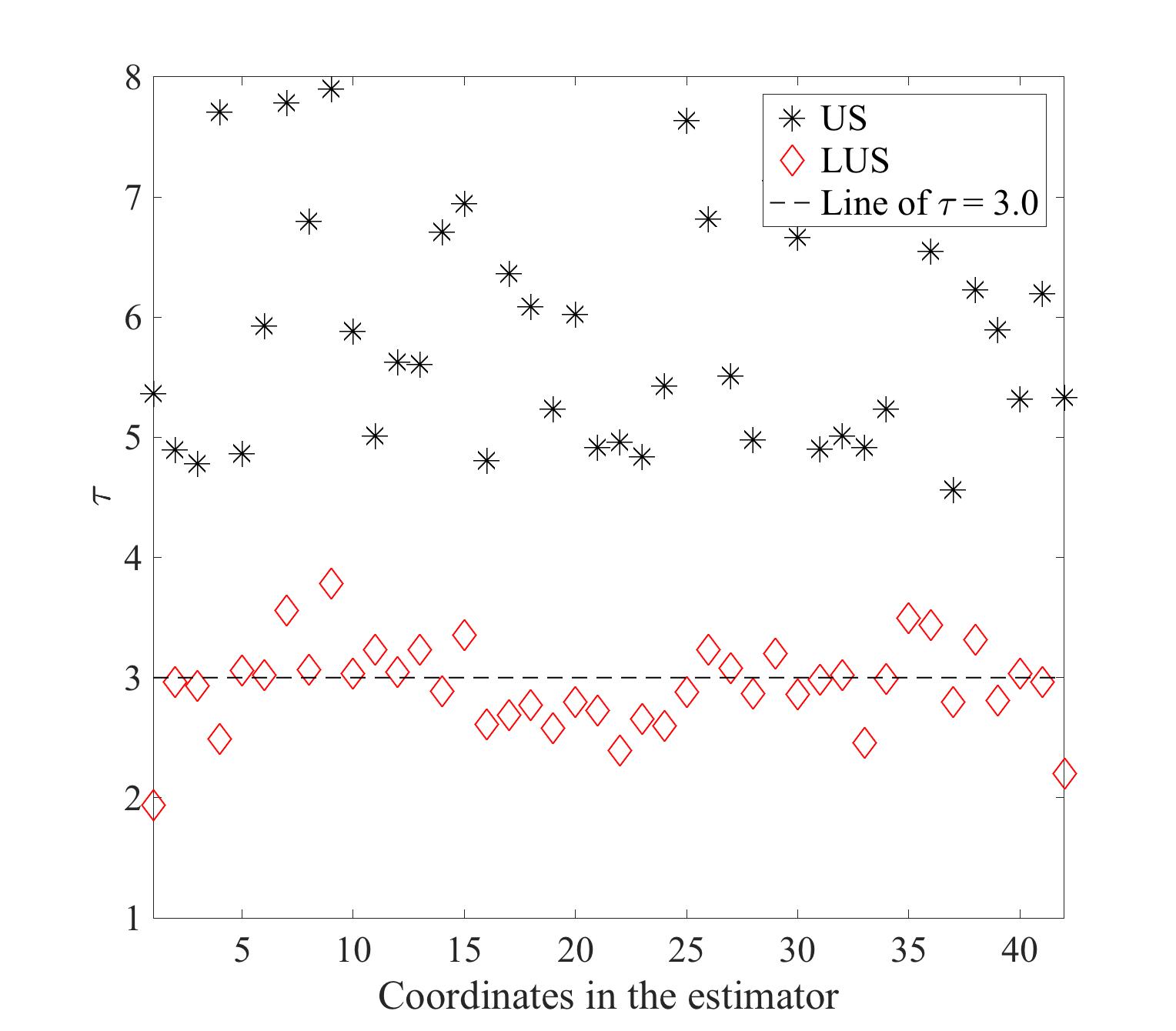}
}\vskip -0.1in
\caption{The $\tau$ value for each coordinate under different values of $\gamma$. $\tau=\text{Var}(\hat{\theta}_{Sub}) / \text{Var}(\hat{\theta}_{full})$.}
\label{fig:toy_balance1}
\end{figure*}

\begin{figure*}[!ht]
\centering
\subfigure[Average $\tau$ v.s. $\gamma$]{
\includegraphics[scale=0.38]{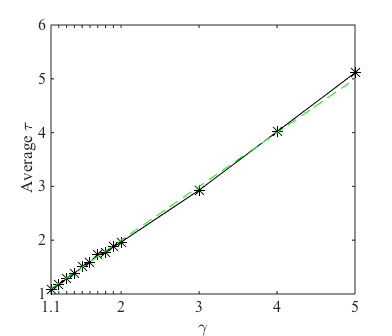}
}\hskip -0.15in
\subfigure[$n_{Sub}/n$ v.s. $\gamma$]{
\includegraphics[scale=0.38]{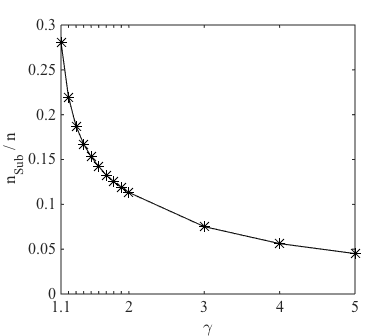}
}\hskip -0.1in
\subfigure[Accuracy v.s. $\frac{n_{Sub}+n_{pilot}}{n}$]{
\includegraphics[scale=0.38]{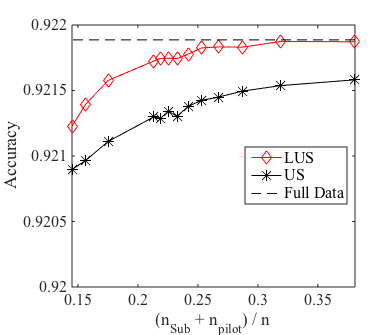}
}\vskip -0.1in
\caption{Plots in the second simulation.}
\label{fig:toy_balance2}
\end{figure*}

The $\tau$ value for each coordinate when $\gamma=\{1.1,2,3\}$ is shown in Fig. \ref{fig:toy_balance1}. The relationship between the average $\tau$ for all the coordinates and $\gamma$ is plotted in Fig. \ref{fig:toy_balance2}(a). Fig. \ref{fig:toy_balance2}(b) reports the relationship between $(n_{Sub}+n_{pilot})/n$ and $\gamma$. In Fig. \ref{fig:toy_balance2}(c), we show the relationship between the prediction accuracy on the test data and the proportion of used training data $(n_{Sub}+n_{pilot})/n$. The results are similar to those of the previous simulation and demonstrate the effectiveness of the LUS method under the marginally balanced (but conditionally imbalanced) case. Fig. \ref{fig:toy_balance2}(c) suggests that we only need about $25\%$ of the full data (including those used for computing the pilot estimator) to achieve the same prediction accuracy as that of the full-sample based MLE.

In the simulations, we have fixed the amount of data used for computing the pilot estimate (as $\frac{n_{pilot}}{n}=10\%$). In general, increasing or decreasing this amount will reduce or increase the variance of the pilot estimate and affect the performance of the LUS estimator accordingly. We provide simulations to show how the performance of the LUS estimator changes with respect to $n_{pilot}$ in Appendix \ref{appendix:pilot_amount}, and the results suggest that a small $n_{pilot}$ is sufficient to let LUS achieve good performance. Another trick with respect to the pilot estimate is to use two or more phases for the pilot estimate. That is, we can recursively apply an obtained LUS estimator as the pilot estimator for the next round. Since the variance of the LUS estimator is smaller than that of the random sampling based estimator (when they subsample the same amount of data), using multiple phases will reduce the variance of the final estimator compared with using random sampling completely.

Despite of these flexibilities on choosing the pilot estimate, for simplicity, we will use random sampling and fix the portion of data for the pilot estimate in the following experiments.

\subsection{MNIST Data}
In this section, we evaluate different methods on the MNIST data\footnote{\url{http://yann.lecun.com/exdb/mnist/}}, which is a benchmark dataset in image classification problems and the state-of-the-art results have achieved less than 1\% test error rate on this dataset. Therefore, the classification accuracy of this problem is high. Note that different from the LCC sampling, our LUS method can handle general logistic models. In this experiment, we let the model function $f$ for the LUS estimator to be one of the state-of-the-art deep neural networks. Since we have no knowledge about the underlying true model in this real dataset, the adopted neural network might be misspecified. In order to save computational cost, we use a simpler neural net structure to obtain the pilot estimate $\tilde{p}(\x,k)$. It is worth mentioning that this simpler neural network is different from the one used for the final LUS estimator, because the LUS method only requires a rough estimate $\tilde{p}(\x,k)$ instead of an explicit estimator which is needed by post-estimation correction based methods. The detailed network structures and parameter settings are provided in Appendix~\ref{appendix:neural}. For the US method, we apply the same network structure used by the final LUS estimator to achieve fair comparison. Since the MNIST data is marginally balanced, the CC method performs the same with the US method and we omit its comparison here.

The training set consists of 60,000 images and the test set has 10,000 images. We uniformly select $n_{pilot}=6000$ data points (i.e., 10\% of the training data) to compute the rough estimate $\tilde{p}(\x,k)$ in \emph{every} repetition and perform 10 repetitions of the experiment to obtain the average performance of different methods. Similar to the simulations, we assume the LUS method samples a number of $n_{Sub}$ data points and let the US method sample $n_{Sub}+n_{pilot}$ data points.
%for fair comparison.
Note that the setting is a bit unfair for LUS because the $n_{pilot}$ data points used for computing $\tilde{p}(\x,k)$ are processed by a simpler neural network. 
However, we still keep this protocol in the experiment.

\begin{figure*}[!ht]
\centering
\subfigure[$n_{Sub}/n$ v.s. $\gamma$]{
\includegraphics[scale=0.15]{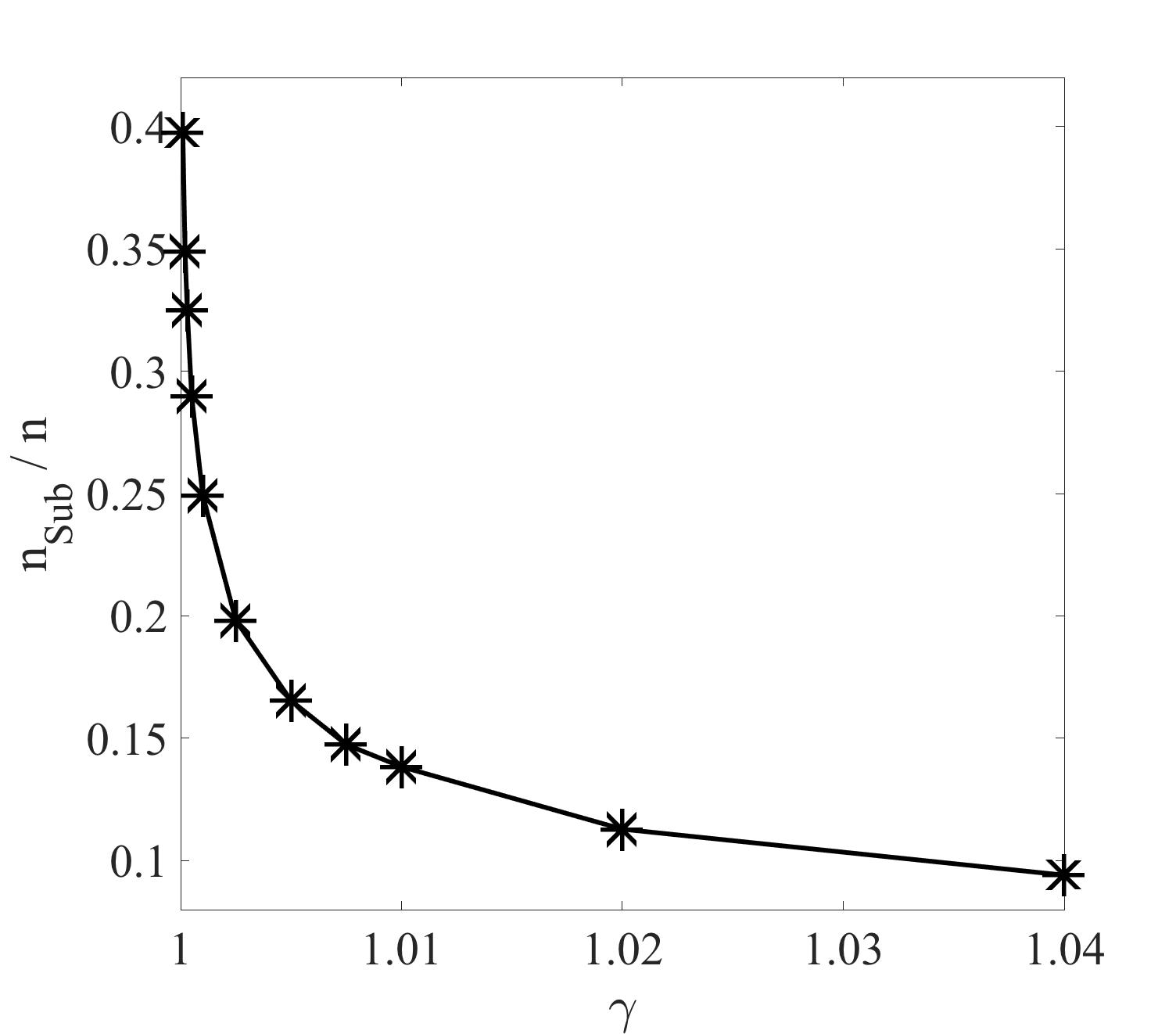}
}\hskip -0.1in
\subfigure[Test error v.s. $\frac{n_{Sub}+n_{pilot}}{n}$]{
\includegraphics[scale=0.15]{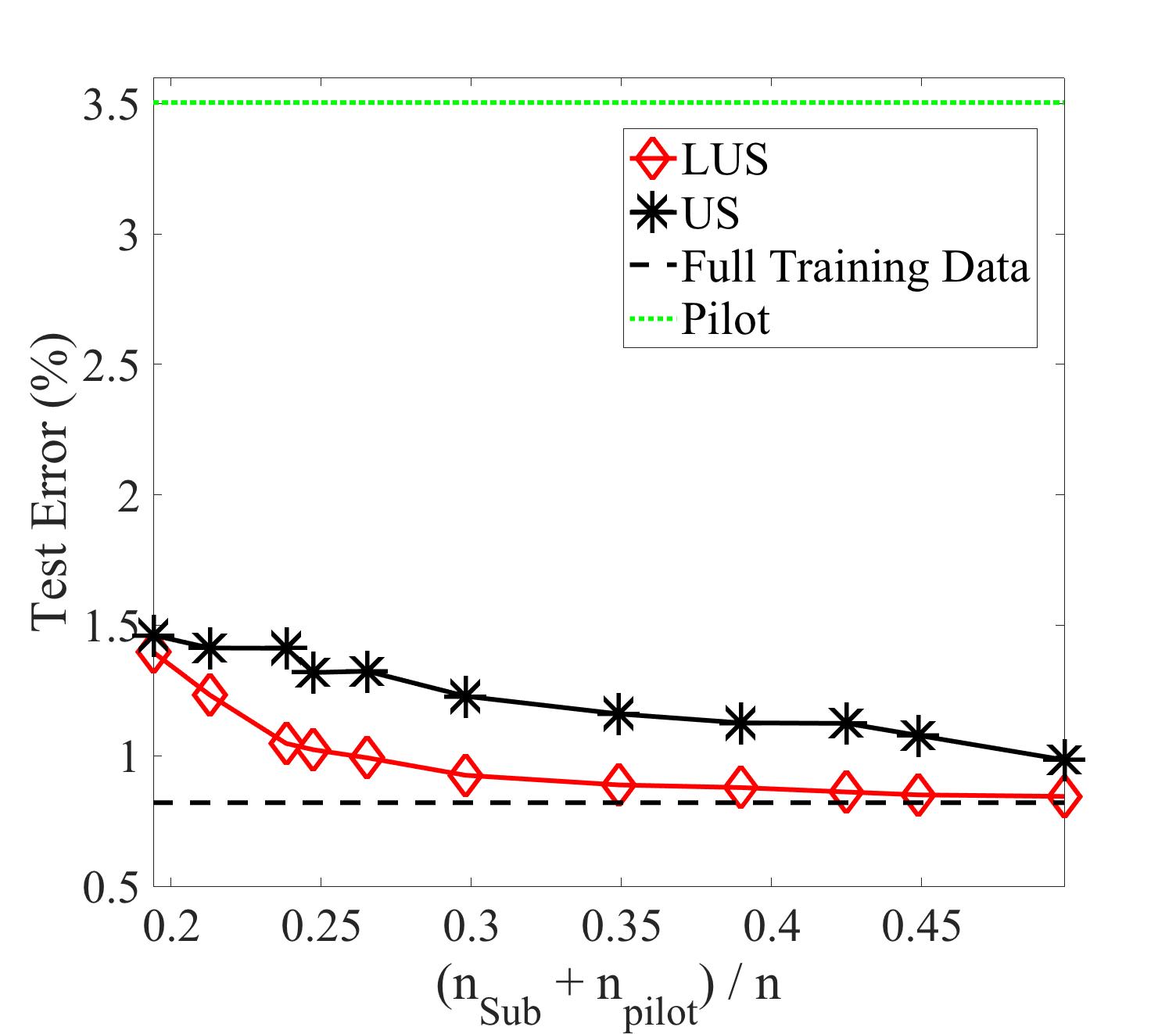}
}
\caption{Plots in MNIST data.}
\label{fig:mnist}
\end{figure*}

We test a number of values of $\gamma$ in the range $(1, 1.04]$ and Fig. \ref{fig:mnist}(a) plots the proportion of used data $(n_{Sub}+n_{pilot})/n$ against $\gamma$. Fig. \ref{fig:mnist}(c) shows the relationship between the test error (\%) and $(n_{Sub}+n_{pilot})/n$.
Note that the rough estimate has a relatively large error rate of about 3.5\%; this is due to the fact that it employs a simpler network structure to save computational cost. Nevertheless, the LUS method can achieve an error rate below 1\%
using only about 25\% of the training data (10\% for the pilot estimate and 15\% for LUS); with about 45\% of the training data (10\% for the pilot estimate and 35\% for LUS), it achieves the same error rate as that obtained by using the full training data. The LUS method consistently outperforms the US method. Table \ref{tb:mnist} shows the speedup of the LUS method compared to the full-sample based estimation.

 \begin{table}[!ht]
 \centering
 \caption{Speedup of the LUS method on MNIST data when using 45\% of the training set (10\% for the pilot estimate and 35\% for LUS), i.e., achieving the same error rate with the full-training-sample based MLE.}
 \begin{tabular}{c|cccccccccc}
   \hline
    & The pilot estimate   & LUS  & Full training data  & Speedup \\
   \hline
Seconds & 51.0 & 369.2 & 1115.1 & 2.7 \\
   \hline
 \end{tabular}
 \label{tb:mnist}
 \end{table}

\subsection{Web Spam Data: Binary Classification}

In this section, we compare the LUS method with the LCC method on the Web Spam data\footnote{\url{http://www.cc.gatech.edu/projects/doi/WebbSpamCorpus.html}}, which is a binary classification problem used in \cite{fithian2014local} to evaluate the LCC method. Since the comparison among the LCC, US and CC methods on this dataset has been reported in \cite{fithian2014local}, we do not repeat them here and focus on the comparison between the LUS and LCC methods. The Web Spam data contains 350,000 web pages and about 60\% of them are web spams. This data set is approximately marginally balanced, but it has been shown to have strong conditional imbalance in \cite{fithian2014local}. We adopt the same settings as described in \cite{fithian2014local} to compare the LUS and the LCC methods. That is, we use linear logistic model and select 99 features which appear in at least 200 documents, and the features are log-transformed. Note that $10\%$ of the observations are uniformly selected to obtain a pilot estimator as in \cite{fithian2014local}. Since we only have a single dataset, we follow \cite{fithian2014local} to uniformly subsample 100 datasets, each of which contains 100,000 data points, as 100 independent `full' datasets, and then repeat the experiments 100 times for comparison.
\begin{figure*}[!ht]
\centering
\subfigure[$\gamma=1.1$]{
\includegraphics[scale=0.25]{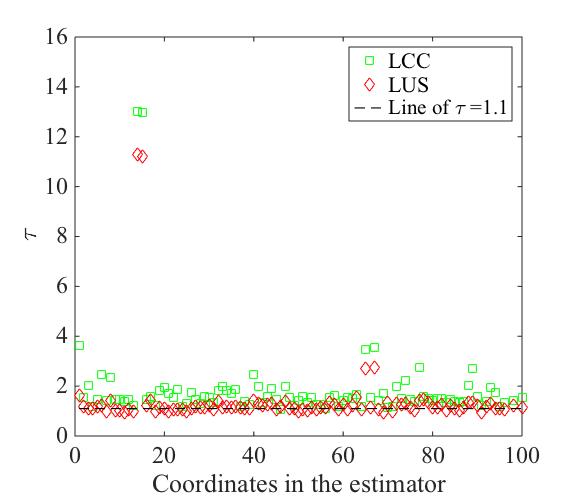}
}\hskip -0.2in
\subfigure[$\gamma=1.2$]{
\includegraphics[scale=0.25]{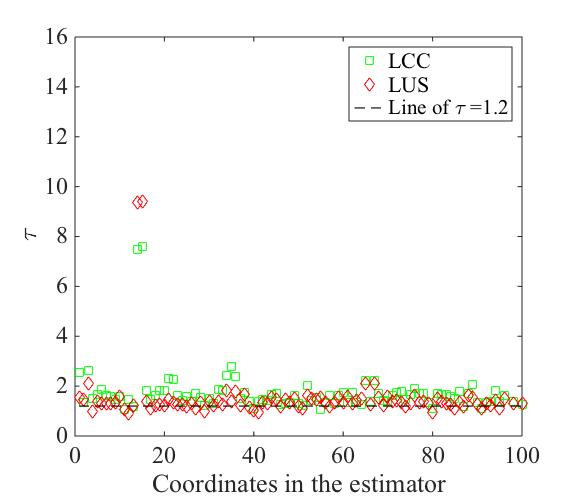}
}\hskip -0.2in
\subfigure[$\gamma=1.3$]{
\includegraphics[scale=0.25]{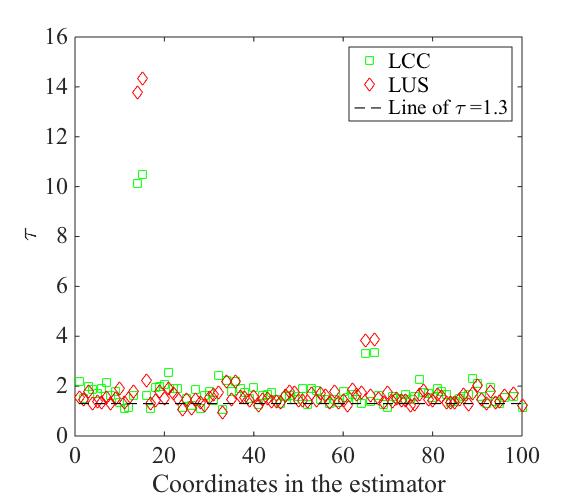}
}\vskip -0.1in

\subfigure[$\gamma=1.4$]{
\includegraphics[scale=0.25]{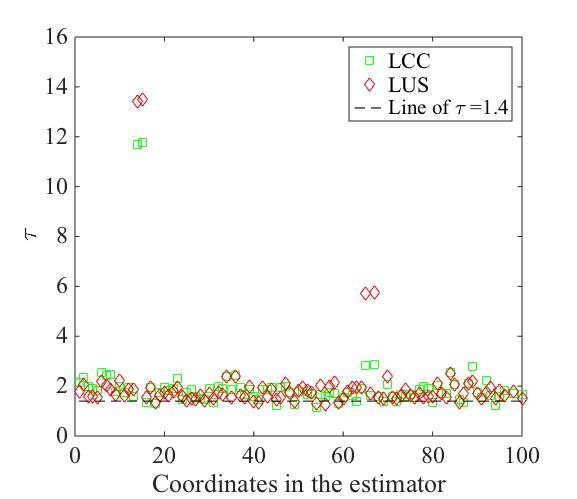}
}\hskip -0.2in
\subfigure[$\gamma=1.5$]{
\includegraphics[scale=0.25]{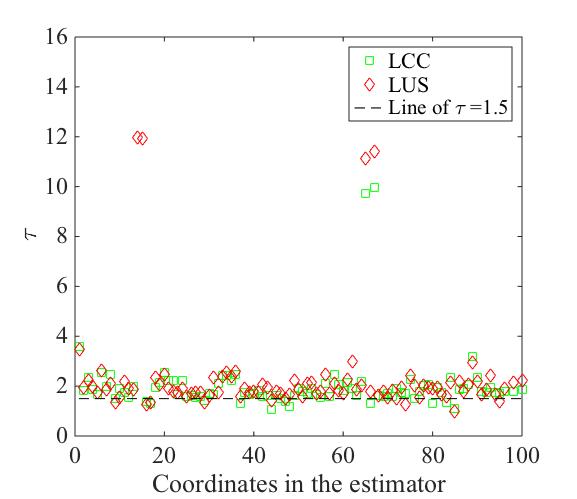}
}\hskip -0.2in
\subfigure[$\gamma=1.6$]{
\includegraphics[scale=0.25]{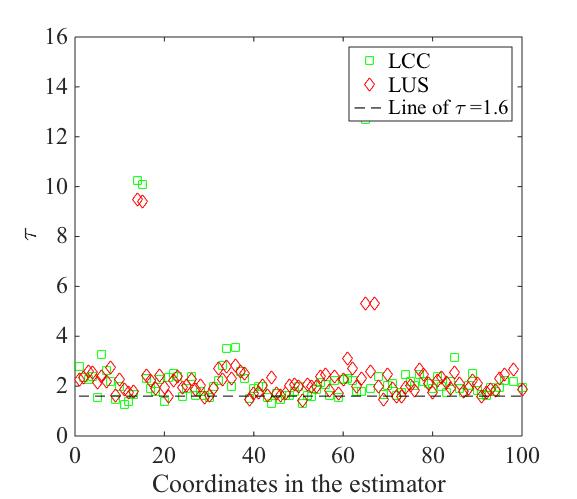}
}\vskip -0.1in

\subfigure[$\gamma=1.7$]{
\includegraphics[scale=0.25]{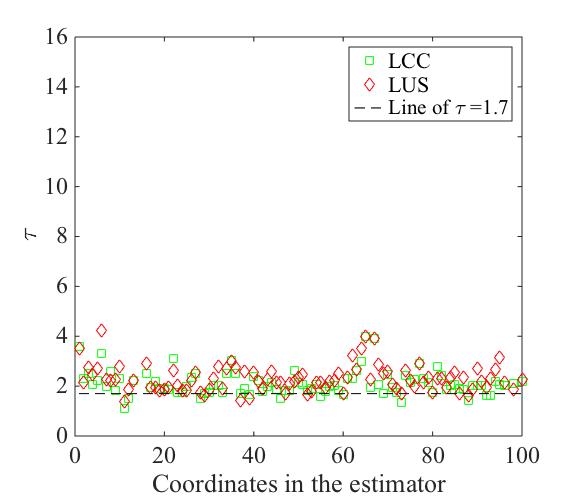}
}\hskip -0.2in
\subfigure[$\gamma=1.8$]{
\includegraphics[scale=0.25]{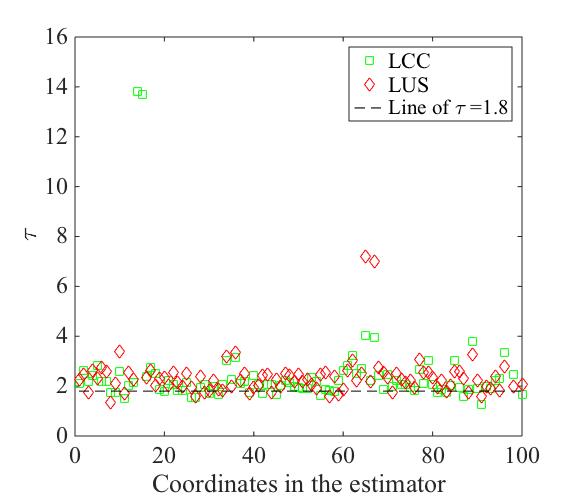}
}\hskip -0.2in
\subfigure[$\gamma=1.9$]{
\includegraphics[scale=0.25]{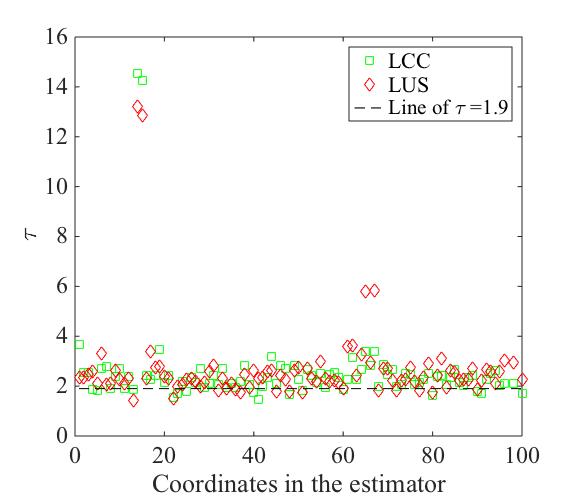}
}\vskip -0.1in

\caption{The $\tau$ value for each coordinate with respect to different values of $\gamma$ in the Web Spam data, where $\tau=\text{Var}(\hat{\theta}_{sub}) / \text{Var}(\hat{\theta}_{full})$.}
\label{fig:webspam1}
\end{figure*}

\begin{figure*}[!ht]
\centering
\subfigure[$n_{Sub}/n$ v.s. $\gamma$]{
\includegraphics[scale=0.4]{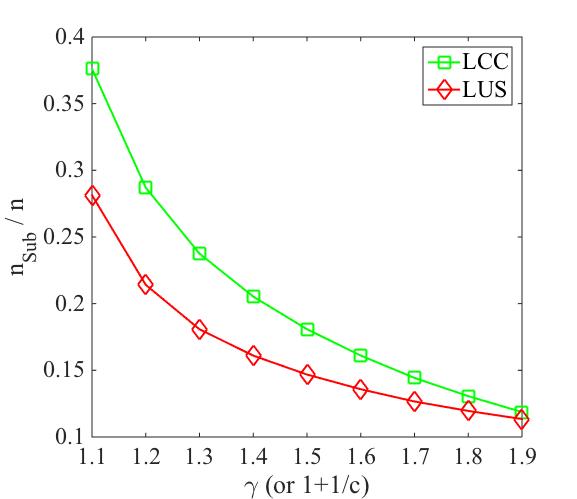}
}\vskip -0.1in

\caption{$n_{Sub}/n$ v.s. $\gamma$ (or $1+\frac{1}{c}$) in Web Spam data.}
\label{fig:webspam2}
\end{figure*}

Observe that when $\gamma\geq2$, the LUS and LCC methods are equivalent to each other by setting the parameter $c=\frac{1}{\gamma-1}\leq1$ of LCC in Section 3.3 of \cite{fithian2014local}. Therefore, we only focus on the case of $1\leq\gamma<2$. In the first experiment, similar to previous experiments, we test different values of $\gamma=\{1.1,1.2,\ldots,1.9\}$ and accordingly set $c=\{10, 5, \ldots, \frac{10}{9}\}$, so that the two methods have the same asymptotic variance. Then, we will compare the number of subsampled data points to see which method is more effective in terms of subsampled data size $n_{Sub}$.

Fig. \ref{fig:webspam1} plots the $\tau$ values for different choices of $\gamma$. As expected from the theoretical results, both the LUS and LCC methods have the same variance that is approximately $\gamma$ (or $1+\frac{1}{c}$) times variance of the full-sample based MLE. 
Next, we compare the subsampling proportion of different methods when $\gamma$ changes in Fig. \ref{fig:webspam2}. From the figure, the LUS method consistently subsamples fewer data points compared with LCC when they achieve the same variance as shown in Fig. \ref{fig:webspam1}.

Alternatively, we fix the proportion of the sampled examples $n_{Sub}/n$ for both LUS and LCC methods (by carefully choosing $\gamma$), and we test the variance of the estimators to see which one has lower variance. Table \ref{tb:webspam} shows the average variance of the coordinates in the LCC estimator and the LUS estimator. We observe that the LUS estimator always achieves lower variance compared to that of the LCC estimator.
These results demonstrate that the LUS method is not only theoretically better justified but also more effective than the LCC meethod in practice (for the case of $\gamma\in[1,2)$).

\begin{table}[!ht]
 \centering
 \caption{Average variance of the coordinates in the estimators from LCC and LUS.}
 \begin{tabular}{c|cccccccccc}
   \hline
$n_{Sub}/n$  & 0.15   & 0.2  & 0.25  & 0.3 \\
   \hline
LCC & 2.2511$\pm$1.9586 & 1.8691$\pm$1.3655 & 1.8270$\pm$1.4078 & 1.7303$\pm$1.1001 \\
LUS & 2.1108$\pm$2.0495 & 1.5608$\pm$1.1668 & 1.6017$\pm$1.5674 & 1.6196$\pm$1.7184 \\
   \hline
 \end{tabular}
 \label{tb:webspam}
 \end{table}

\section{Conclusion}
This paper introduced a general subsampling method for solving large-scale multi-class logistic regression problems. We investigated the asymptotic variance of the proposed estimator when the model is correctly specified. Based on the theoretical analysis, we proposed an effective sampling strategy called Local Uncertainty Sampling to achieve any given level of desired variance. We proved that the method always achieves lower variance than random subsampling for a given expected sample size, and the improvement may be significant under the favorable condition of strong conditional imbalance. Therefore the method can effectively accelerate the computation of large-scale multi-class logistic regression in practice. 
Empirically, we will need a pilot estimate of the probability to setup the acceptance probability. We proved that the variance of the proposed estimator remains a constant which is independent of the randomness of this pilot estimate as long as it is consistent.

We also studied the case of model misspecification. We showed that for binary classification problems ($K=2$), the proposed method can generate a consistent estimator (to the best estimator of the misspecified model) if the pilot estimator is consistent. For $K>2$, the proposed estimator is biased and we also provided analysis to quantify the bias.

The empirical studies support the theory and demonstrate that the local uncertainty sampling method outperforms the uniform sampling, case-control sampling and the local case-control sampling methods under various settings. By using the proposed method, we are able to select a very small subset of the original data to achieve the same performance as that of  the full dataset, which provides an effective mean for big data computation under limited resources. 

This work suggests several future directions. First, as we have mentioned at the end of Section \ref{sec:simulation2}, one can iteratively apply an obtained LUS estimator as the pilot estimator for the next round of fitting the model. This is closely related to the boosting method \cite{friedman2000additive} and a deep discussion on this relationship would be of great interest. Second, considering a situation of online learning with limited budget, the LUS method would likely provide an effective sampling strategy for this problem. Moreover, in high-dimensional settings, where sparse models, e.g., Lasso and Group Lasso, are widely adopted, it would be interesting to extend the LUS method to deal with regularized multi-class logistic regression with special considerations on high-dimensional asymptotic regime.

%\section*{Acknowledgements}
%We would like to thank the reviewers for their valuable comments.
%This research was partially supported by NSF IIS-1250985, NSF IIS-1407939, NSF DMS-1811315, and NIH R01AI116744.

\bibliographystyle{unsrt}
\bibliography{paper}

\begin{thebibliography}{10}

\bibitem{cortes2008sample}
Corinna Cortes, Mehryar Mohri, Michael Riley, and Afshin Rostamizadeh.
\newblock Sample selection bias correction theory.
\newblock In {\em Algorithmic Learning Theory}, pages 38--53. Springer, 2008.

\bibitem{cortes2010learning}
Corinna Cortes, Yishay Mansour, and Mehryar Mohri.
\newblock Learning bounds for importance weighting.
\newblock In {\em Advances in Neural Information Processing Systems}, pages
  442--450, 2010.

\bibitem{dhillon2013new}
Paramveer Dhillon, Yichao Lu, Dean~P Foster, and Lyle Ungar.
\newblock New subsampling algorithms for fast least squares regression.
\newblock In {\em Advances in Neural Information Processing Systems}, pages
  360--368, 2013.

\bibitem{mineiro2013loss}
Paul Mineiro and Nikos Karampatziakis.
\newblock Loss-proportional subsampling for subsequent erm.
\newblock {\em arXiv preprint arXiv:1306.1840}, 2013.

\bibitem{xie1989logit}
Yu~Xie and Charles~F Manski.
\newblock The logit model and response-based samples.
\newblock {\em Sociological Methods and Research}, 17(3):283--302, 1989.

\bibitem{zadrozny2004learning}
Bianca Zadrozny.
\newblock Learning and evaluating classifiers under sample selection bias.
\newblock In {\em Proceedings of the International Conference on Machine
  Learning}, page 114. ACM, 2004.

\bibitem{zhang2000value}
Tong Zhang and F~Oles.
\newblock The value of unlabeled data for classification problems.
\newblock In {\em Proceedings of the International Conference on Machine
  Learning}, pages 1191--1198. Citeseer, 2000.

\bibitem{chawla2004editorial}
Nitesh~V Chawla, Nathalie Japkowicz, and Aleksander Kotcz.
\newblock Editorial: special issue on learning from imbalanced data sets.
\newblock {\em ACM Sigkdd Explorations Newsletter}, 6(1):1--6, 2004.

\bibitem{he2009learning}
Haibo He and Edwardo~A Garcia.
\newblock Learning from imbalanced data.
\newblock {\em IEEE Transactions on Knowledge and Data Engineering},
  21(9):1263--1284, 2009.

\bibitem{abe2004iterative}
Naoki Abe, Bianca Zadrozny, and John Langford.
\newblock An iterative method for multi-class cost-sensitive learning.
\newblock In {\em Proceedings of the ACM SIGKDD International Conference on
  Knowledge Discovery and Data Mining}, pages 3--11. ACM, 2004.

\bibitem{kim2002pattern}
Hyun-Chul Kim, Shaoning Pang, Hong-Mo Je, Daijin Kim, and Sung~Yang Bang.
\newblock Pattern classification using support vector machine ensemble.
\newblock In {\em Proceedings of the International Conference on Pattern
  Recognition}, volume~2, pages 160--163. IEEE, 2002.

\bibitem{widodo2007support}
Achmad Widodo and Bo-Suk Yang.
\newblock Support vector machine in machine condition monitoring and fault
  diagnosis.
\newblock {\em Mechanical Systems and Signal Processing}, 21(6):2560--2574,
  2007.

\bibitem{tan2003multi}
Aik~Choon Tan, David Gilbert, and Yves Deville.
\newblock Multi-class protein fold classification using a new ensemble machine
  learning approach.
\newblock {\em Genome Informatics}, 14:206--217, 2003.

\bibitem{lecun1998gradient}
Yann LeCun, L{\'e}on Bottou, Yoshua Bengio, and Patrick Haffner.
\newblock Gradient-based learning applied to document recognition.
\newblock {\em Proceedings of the IEEE}, 86(11):2278--2324, 1998.

\bibitem{fithian2014local}
William Fithian and Trevor Hastie.
\newblock Local case-control sampling: efficient subsampling in imbalanced data
  sets.
\newblock {\em Annals of Statistics}, 42(5):1693, 2014.

\bibitem{webb2006introducing}
Steve Webb, James Caverlee, and Calton Pu.
\newblock Introducing the webb spam corpus: using email spam to identify web
  spam automatically.
\newblock In {\em Proceedings of the Third Conference on Email and Anti-Spam},
  2006.

\bibitem{mantel1959statistical}
Nathan Mantel and William Haenszel.
\newblock Statistical aspects of the analysis of data from retrospective
  studies.
\newblock {\em Journal of the National Cancer Institute}, 22(4):719--748, 1959.

\bibitem{breslow1982design}
Norman Breslow.
\newblock Design and analysis of case-control studies.
\newblock {\em Annual Review of Public Health}, 3(1):29--54, 1982.

\bibitem{anderson1972separate}
James~A Anderson.
\newblock Separate sample logistic discrimination.
\newblock {\em Biometrika}, 59(1):19--35, 1972.

\bibitem{king2001logistic}
Gary King and Langche Zeng.
\newblock Logistic regression in rare events data.
\newblock {\em Political Analysis}, 9(2):137--163, 2001.

\bibitem{horvitz1952generalization}
Daniel~G Horvitz and Donovan~J Thompson.
\newblock A generalization of sampling without replacement from a finite
  universe.
\newblock {\em Journal of the American statistical Association},
  47(260):663--685, 1952.

\bibitem{scott1986fitting}
Alastair~J Scott and CJ~Wild.
\newblock Fitting logistic models under case-control or choice based sampling.
\newblock {\em Journal of the Royal Statistical Society. Series B
  (Methodological)}, pages 170--182, 1986.

\bibitem{scott1991fitting}
AJ~Scott and CJ~Wild.
\newblock Fitting logistic regression models in stratified case-control
  studies.
\newblock {\em Biometrics}, pages 497--510, 1991.

\bibitem{scott2002robustness}
Alastair Scott and Chris Wild.
\newblock On the robustness of weighted methods for fitting models to
  case--control data.
\newblock {\em Journal of the Royal Statistical Society: Series B (Statistical
  Methodology)}, 64(2):207--219, 2002.

\bibitem{atkeson1997locally}
Christopher~G Atkeson, Andrew~W Moore, and Stefan Schaal.
\newblock Locally weighted learning for control.
\newblock In {\em Lazy Learning}, pages 75--113. Springer, 1997.

\bibitem{friedman2000additive}
Jerome Friedman, Trevor Hastie, Robert Tibshirani, et~al.
\newblock Additive logistic regression: a statistical view of boosting (with
  discussion and a rejoinder by the authors).
\newblock {\em The annals of statistics}, 28(2):337--407, 2000.

\bibitem{vedaldi2015matconvnet}
Andrea Vedaldi and Karel Lenc.
\newblock Matconvnet: Convolutional neural networks for matlab.
\newblock In {\em Proceedings of the Annual ACM Conference on Multimedia
  Conference}, pages 689--692. ACM, 2015.

\end{thebibliography}

\clearpage

\appendix

\section{Proofs}\label{appendix:proof}
We first state a lemma that will be useful for the proofs.
\begin{lem}\label{lem:1}
For some norm $\|\cdot\|$ defined on the parameter space of $\bT$, assume that the quantities $\|\nabla_{\bt_k} f(\x,\bt_k)\|$, $\|\nabla_{\bt_k}^2 f(\x,\bt_k)\|$ and $\|\nabla_{\bt_k}^3 f(\x,\bt_k)\|$ for $k=1,\ldots,K-1$ are bounded. Then, for any compact set $\mathbb{S}\in\R^{d(K-1)}$, we have
{\small
\begin{align*}
	&\sup_{\bT\in\mathbb{S}}|\hat{R}_n(\bT)-R(\bT)|\xrightarrow{p}0,\ \ \
	\sup_{\bT\in\mathbb{S}}\|\nabla_{\bT} \hat{R}_n(\bT)-\nabla_{\bT} R(\bT)\|\xrightarrow{p}0,
	\\
	&\sup_{\bT\in\mathbb{S}}\|\nabla_{\bT}^2 \hat{R}_n(\bT)-\nabla_{\bT}^2 R(\bT)\|\xrightarrow{p}0.
\end{align*}
}
\end{lem}
\begin{proof}
For a fixed $\bT$, we define
{\small
\begin{align*}
\psi(\x, y, \bT)=\sum_{k=1}^{K-1}\mathbb{I}(y=k)\cdot g(\x,\bt_k)-\log\left(1+\sum_{k=1}^{K-1}e^{g(\x,\bt_k)}\right).
\end{align*}
}\noindent
Then, we have $\hat{R}_n(\bT)=\frac{1}{n}\sum_{i=1}^n z_i\psi(\x_i,y_i,\bT)$ and
$R(\bT)=\E_{\A} \Big[z_i\psi(\x_i,y_i,\bT)\Big]$. By the Law of Large Numbers, we know that $\hat{R}_n(\bT)$ converges point-wisely to $R(\bT)$ in probability.

According to the assumption, there exists a constant $M>0$ such that
{\small
\begin{align*}
\|\nabla_{\bT}\psi(\x,y,\bT)\|
&\leq \sum_{k=1}^{K-1} \left\| \left(\mathbb{I}(y=k)-\frac{e^{g(\x,\bt_k)}}{1+\sum_{k'=1}^{K-1}e^{g(\x,\bt_{k'})}}\right) \nabla_{\bt_k} f(\x,\bt_k) \right\|\leq M.
%&\leq \sum_{k=1}^{K-1}\left\| \nabla_{\bt_k} f(\x,\bt_k) \right\| \leq M.
\end{align*}
}\noindent
Given any $\epsilon>0$, we may find a finite cover $\mathbb{S}_\epsilon =\{\bT_1,\ldots,\bT_T\} \subset \mathbb{S}$ so that for any $\bT\in\mathbb{S}$, there exists $\bT'\in\mathbb{S}_\epsilon$ such that $|\psi(\x,y,\bT)-\psi(\x,y,\bT')| < \epsilon$. Since $\mathbb{S}_\epsilon$ is finite, as $n\to\infty$, $\sup_{\bT\in\mathbb{S}_\epsilon} | \hat{R}_n(\bT)-R(\bT)|$ converges to $0$ in probability. Therefore as $n\to\infty$, with probability $1$, we have
{\small
\[
\sup_{\bT \in \mathbb{S}} | \hat{R}_n(\bT)-R(\bT)|< 2 \epsilon + \sup_{\bT \in \mathbb{S}_\epsilon} |\hat{R}_n(\bT)-R(\bT)| \to 2 \epsilon .
\]
}\noindent
Let $\epsilon \to 0$, we obtain the first bound. The second and the third bounds can be obtained similarly.
\end{proof}

\subsection{Proof of Theorem \ref{thm:variance}}
\begin{proof}
For notational simplicity, throughout the proofs we will abbreviate the point-wise functions
$f(\x,\bt_k)$, $g(\x,\bt_k)$, $p(\x,k)$, $a(\x,k)$ and $q(\x)$ at point $\x$ as $f_k$, $g_k$, $p_k$, $a_k$ and $q$, respectively.
Moreover, let $\nabla_i=\nabla_{\bt_i}f(\x,\bT^0)$.

(1) Define $\bm{g}=(g_1,\ldots,g_{K-1})^\top$ and the convex function
{\small
\[
G(\bm{g})=\log\left(1+\sum_{k=1}^{K-1}e^{g_k}\right).
\]
}\noindent
The Bregman divergence of the convex function $G(\bm{g})$ is 
{\small
\begin{align*}
\Delta(\bm{g},\bm{g}')&=G(\bm{g})-G(\bm{g}')-\nabla G(\bm{g}')^{\top}(\bm{g}-\bm{g}')
\\
&=\log\frac{1+\sum_{k=1}^{K-1}e^{g_k}}{1+\sum_{k=1}^{K-1}e^{g_k'}}+
\sum_{k=1}^{K-1} \frac{e^{g_k'}}{1+\sum_{k=1}^{K-1}e^{g_k'}} (g_k'-g_k).
\end{align*}
}\noindent
Since $G(\cdot)$ is convex, we have $\Delta(\bm{g},\bm{g}')\geq0$ and $\Delta(\bm{g},\bm{g}')=0$ only when $\bm{g}=\bm{g}'$. 

Given any $\bm{h}(\x)$, if we define
{\small
\[
\tilde{R}\left(\bm{h}(\x)\right)=\E_{\x,y,z\sim\A}\ z\left[\sum_{k=1}^{K-1}\mathbb{I}(y=k)h_k(\x)-\log\left(1+\sum_{k=1}^{K-1}e^{h_k(\x)}\right) \right],
\]
}\noindent
then $\bm{g}(\x,\bT^0)$ is the global optimizer of $\tilde{R}(\cdot)$ and $R(\bT)=\tilde{R}(\bm{g}(\x,\bT^0))$. It implies that,
{\small
\begin{align*}
&R(\bT^0)-R(\bT)=
\E_{\x,y\sim\D}\ a(\x,y)\Delta(\bm{g}(\x,\bT),\bm{g}(\x,\bT^0)).
\end{align*}
}\noindent

The assumption of the theorem implies that the parameter space is compact and we have $\P_{\D}(f(\x,\bT) \neq f(\x,\bT^0)) >0$, $\forall\bT\neq\bT^0$.  Therefore, for any $\bT\neq\bT^0$, we have $R(\bT) < R(\bT^0)$. 

It follows that given any $\epsilon'>0$, there exists $\epsilon>0$ so that $R(\bT) \geq R(\bT^0)-2 \epsilon$ implies that $\|\bT-\bT^0\|<\epsilon'$. Now according to Lemma \ref{lem:1}, given any $\delta >0$, when $n \to \infty$, with probability larger than $1-\delta$ we have
{\small
\[
R(\hat{\bT}_{Sub}) \geq \hat{R}_n(\hat{\bT}_{Sub}) - \epsilon
\geq  \hat{R}_n(\bT^0) - \epsilon
\geq  R(\bT^0) - 2 \epsilon .
\]
}\noindent
This implies that $\|\hat{\bT}_{Sub}-\bT^0\|<\epsilon'$.

(2) Let $\bT=(\bt_1^{\top},\ldots,\bt_{K-1}^{\top})^{\top}\in\mathbb{R}^{d(K-1)}$. By the mean value theorem, we have
{\small
\begin{align*}
\sqrt{n}\left(\hat{\bT}_{Sub}-\bT^0 \right)=-\nabla_{\bT}^2\hat{R}_n(\bar{\bT})^{-1}\sqrt{n}\nabla_{\bT}\hat{R}_n(\bT^0) ,
\end{align*}
}\noindent
where $\bar{\bT}=t\bT^0+(1-t)\hat{\bT}_{Sub}$ for some $t\in[0,1]$. 
Since $\hat{\bT}_{Sub}\xrightarrow{p}\bT^0$, we have $\bar{\bT}\xrightarrow{p}\bT^0$.
By the continuous mapping theorem, $\nabla_{\bT}^2 \hat{R}_n(\bar{\bT})^{-1}$ converges to $\nabla_{\bT}^2 \hat{R}_n(\bT^0)^{-1}$ which will eventually converge to $\nabla_{\bT}^2 R(\bT^0)^{-1}$ in probability according to Lemma~\ref{lem:1}. By the Slutsky's Theorem, the limiting distribution of $\sqrt{n}\left(\hat{\bT}_{Sub}-\bT^0\right)$ is given by
{\small
\[
-\nabla_{\bT}^2 R(\bT^0)^{-1}\sqrt{n}\nabla_{\bT}\hat{R}_n(\bT^0) .
\]
}\noindent
Observe that $\sqrt{n}\nabla_{\bT}\hat{R}_n(\bT^0)$ is the sum of $n$ i.i.d. random variables with mean $\E \sqrt{n}\nabla_{\bT}\hat{R}_n(\bT^0)=\sqrt{n}\E \nabla_{\bT}\hat{R}_n(\bT^0)=\bm{0}$. The variance of $\sqrt{n}\left(\hat{\bT}_{Sub}-\bT^0\right)$ is
%%%%%%%%%%%%%%%%%%%%%%%%%%%%%%%%%%%%%%%%%%%%%%%%%%%%%%%
{\small
\begin{align*}
	Var\left(\sqrt{n}(\hat{\bT}_{Sub}-\bT^0)\right)=
	\H(\bT^0)^{-1} Var\left(\sqrt{n}\nabla_{\bT}\hat{R}_n(\bT^0)\right)\H(\bT^0)^{-1},
\end{align*}
}\noindent
where $\H(\bT^0)=-\nabla_{\bT}^2R(\bT^0)$. Now, we derive explicit formula for $\H(\bT^0)$:
{\small
\begin{align*}
	\H(\bT^0)
	=\left[
	\begin{array}{cccc}
		\H_{11} & \H_{12} & \cdots & \H_{1,K-1} \\
		\H_{21} & \H_{22} & \cdots & \H_{2,K-1} \\
		\vdots & \vdots & \ddots & \vdots \\
		\H_{K-1,1} & \H_{K-1,2} & \cdots & \H_{K-1,K-1} \\
	\end{array}
	\right],
\end{align*}
}\noindent
where
{\small
\begin{align*}
\H_{jj}=\ &-\nabla_{jj}^2{R}(\bT^0)
\\
=\ &\E_{\x,y,z \sim \A}\
Z\left(
\frac{e^{g_j^0}}{1+\sum_{k=1}^{K-1}e^{g_k^0}}
\right)\left(
1-\frac{e^{g_j^0}}{1+\sum_{k=1}^{K-1}e^{g_k^0}}
\right)\cdot \nabla_j\nabla_j^\top
\nonumber
\\
=\ &\E_{\x,y \sim \D}\
a_y\left(
\frac{e^{g_j^0}}{1+\sum_{k=1}^{K-1}e^{g_k^0}}
\right)\left(
1-\frac{e^{g_j^0}}{1+\sum_{k=1}^{K-1}e^{g_k^0}}
\right)\cdot
\nabla_j\nabla_j^\top
\nonumber
\\
=\ &\E_{\x,y \sim \D}\
a_y\cdot\frac{a_jp_j\sum_{k\neq j}^Ka_kp_k}{\left(\sum_{k=1}^{K}a_kp_k\right)^2}\cdot
\nabla_j\nabla_j^\top
\nonumber
\\
=\ &\E_{\x\sim\D}\
\frac{a_jp_j\sum_{k\neq j}^Ka_kp_k}{\sum_{k=1}^{K}a_kp_k}\cdot
\nabla_j\nabla_j^\top,
\nonumber
\end{align*}
}\noindent
where the last equality is due to $e^{g_k^0}=\frac{a_kp_k}{a_Kp_K}$ based on Eq. (\ref{eq:logodds}),
%%%%%%%%%%%%%%%%%%%%%%%%%%%%%%%%%%%%%%%%%%%%%%%%%%
and
%%%%%%%%%%%%%%%%%%%%%%%%%%%%%%%%%%%%%%%%%%%%%%%%%%
{\small
\begin{align*}
\H_{ij}=\ &-\nabla_{ij}^2{R}(\bT^0)
\\
=\ &-\E_{\x,y \sim \D}\
a_y\left(
\frac{e^{g_i^0}}{1+\sum_{k=1}^{K-1}e^{g_k^0}}
\right)\left(
\frac{e^{g_j^0}}{1+\sum_{k=1}^{K-1}e^{g_k^0}}
\right)\cdot
\nabla_i\nabla_j^\top
\nonumber
\\
=\ &-\E_{\x\sim\D}\
\frac{a_ip_ia_jp_j}{\sum_{k=1}^{K}a_kp_k}\cdot
\nabla_i\nabla_j^\top.
\nonumber
\end{align*}
}\noindent
This implies that we can rewrite $\H(\bT^0)$ as
{\small
\begin{align*}
	\H(\bT^0)=\E_{\x\sim\D}
	\bm{\nabla}
	\mathbf{S}
	\bm{\nabla}^{\top}.
\end{align*}
}\noindent
%%%%%%%%%%%%%%%%%%%%%%%%%%%%%%%%%%%%%%%%%%%%%%%%%%%%

Next, we derive an explicit formula for $Var\left(\sqrt{n}\nabla_{\bT}\hat{R}_n(\bT^0)\right)$ as
{\small
\begin{align*}
Var\left(\sqrt{n}\nabla_{\bT}\hat{R}_n(\bT^0)\right)=
	\left[
	\begin{array}{cccc}
		\V_{11} & \V_{12} & \cdots & \V_{1,K-1} \\
		\V_{21} & \V_{22} & \cdots & \V_{2,K-1} \\
		\vdots & \vdots & \ddots & \vdots \\
		\V_{K-1,1} & \V_{K-1,2} & \cdots & \V_{K-1,K-1} \\
	\end{array}
	\right],
\end{align*}
}\noindent
where
{\small
\begin{align*}
\V_{jj}
=\ &\E_{\x,y,z\sim\A}\ z\left(\mathbb{I}(y=j)-
\frac{e^{g_j^0}}{1+\sum_{k=1}^{K-1}e^{g_k^0}}\right)^2
\cdot\nabla_j\nabla_j^\top
\qquad (z^2=z)
\nonumber
\\
=\ &\E_{\x,y \sim \D}\ a_y\left(\mathbb{I}(y=j)-
\frac{e^{g_j^0}}{1+\sum_{k=1}^{K-1}e^{g_k^0}} \right)^2
\cdot\nabla_j\nabla_j^\top
\nonumber
	\\
=\ &\E_{\x\sim\D}
\left[
\sum_{k\neq j}^Ka_kp_k\cdot\left(\frac{a_jp_j}{\sum_{k=1}^{K}a_kp_k}\right)^2+a_jp_j\left(1-\frac{a_jp_j}{\sum_{k=1}^{K}a_kp_k}\right)^2
\right]
\cdot\nabla_j\nabla_j^\top
\nonumber
\\
=\ &\E_{\x\sim\D}
\frac{a_jp_j\sum_{k\neq j}^Ka_kp_k}{\sum_{k=1}^{K}a_kp_k}\cdot
\nabla_j\nabla_j^\top,
\nonumber
\end{align*}
}\noindent
%%%%%%%%%%%%%%%%%%%%%%%%%%%%%%%%%%%%%%%%%%%%%%%%%%%%
and
%%%%%%%%%%%%%%%%%%%%%%%%%%%%%%%%%%%%%%%%%%%%%%%%%%%%
{\small
\begin{align}
\V_{ij}
=\ &\E_{\x,y\sim\D}\ a_y\left(\mathbb{I}(y=i)-
\frac{e^{g_i^0}}{1+\sum_{k=1}^{K-1}e^{g_k^0}}\right)
\left(\mathbb{I}(y=j)-
\frac{e^{g_j^0}}{1+\sum_{k=1}^{K-1}e^{g_k^0}}\right)\cdot
\nabla_i\nabla_j^\top
\nonumber
\\
=\ &\E_{\x\sim\D}
\left[
\sum_{k\neq i, j}^K\frac{a_i p_i a_j p_j a_k p_k}{\left(\sum_{k=1}^Ka_kp_k\right)^2}
-\frac{a_i p_i a_j p_j\sum_{k\neq i}^Ka_kp_k}{\left(\sum_{k=1}^Ka_kp_k\right)^2}
-\frac{a_i p_i a_j p_j\sum_{k\neq j}^Ka_kp_k}{\left(\sum_{k=1}^Ka_kp_k\right)^2}
\right]
\cdot\nabla_i\nabla_j^\top
\nonumber
\\
=\ &-\E_{\x\sim\D}\
\frac{a_i p_i a_j p_j}{\sum_{k=1}^{K}a_kp_k}\cdot
\nabla_i\nabla_j^\top.
\nonumber
\end{align}
}\noindent
This indicates that $\H(\bT^0)=\text{Var}\left(\sqrt{n}\nabla_{\bT}\hat{R}_n(\bT^0)\right)$. Hence, we have
{\small
\[
\sqrt{n}\left(\hat{\bT}_{Sub}-\bT^0\right)\xrightarrow{d}\mathcal{N}\left(\bm{0},
  	\left[
\E_{\x\sim\D}
	\bm{\nabla}
	\mathbf{S}
	\bm{\nabla}^{\top}
\right]^{-1}
\right).
\]
}\noindent
This proves the desired result.
\end{proof}

\subsection{Proof of Theorem \ref{thm:LUS}}
\begin{proof}
Observe that  the quantity $\bm{\nabla}$ is independent of the sampling probability $a(\x,k)$, and the dependence on $a(\x,k)$ in Eq. (\ref{eq:finalvariance}) comes from $\mathbf{S}$.
Therefore, to prove the desired bound of the variance, we only need to show
{\small
\begin{align*}
\gamma\mathbf{S}\succeq\mathbf{S}_{full}=\gamma\mathbf{S}_{US:\frac{1}{\gamma}} .
\end{align*}
}\noindent
In order to prove this inequality, we only need to verify that for any vector $\bm{\beta}\in\R^{K-1}$, $\bm{\beta}^{\top}(\gamma\mathbf{S}-\mathbf{S}_{full})\bm{\beta}\geq0$. That is,
{\small
\begin{align}
\gamma\sum_{i=1}^{K-1}a_ip_i\beta_i^2-\frac{\gamma}{\sum_{i=1}^Ka_ip_i}\left(\sum_{i=1}^{K-1}a_ip_i\beta_i\right)^2-
\sum_{i=1}^{K-1}p_i\beta_i^2+\left(\sum_{i=1}^{K-1}p_i\beta_i\right)^2\geq0,
\label{eq:PSD}
\end{align}
}\noindent
where $\sum_{i=1}^Kp_i=1$ and $0\leq a_i\leq1$.

%\begin{enumerate}
%\renewcommand{\labelenumi}{(\theenumi)}
%	\item

\medskip
\noindent
(1) Consider the first case in Theorem \ref{thm:LUS}, where $\gamma\geq2q$.
Given $\x$, we consider the following three sub-cases.
\medskip

(1.1) Assume that $q=p_K\geq0.5$. Denote the left side of Eq. (\ref{eq:PSD})  by $F(a)$, and plug Eq. (\ref{eq:a1}) into $F(a)$,  we obtain $\sum_{i=1}^Ka_ip_i=\frac{4}{\gamma}q(1-q)$, and
{\small
\begin{align*}
F(a)&=(2q-1)\sum_{i=1}^{K-1}p_i\beta_i^2-\left(\frac{q}{1-q}-1\right)\left(\sum_{i=1}^{K-1}p_i\beta_i\right)^2
\\
&\geq\frac{2q-1}{1-q}\left(\sum_{i=1}^{K-1}p_i\beta_i\right)^2-\left(\frac{q}{1-q}-1\right)\left(\sum_{i=1}^{K-1}p_i\beta_i\right)^2=0,
\end{align*}
}\noindent
where the inequality is obtained by the Cauchy-Schwartz inequality, which implies
{\small
\[
\sum_{i=1}^{K-1}p_i\beta_i^2\geq\frac{\left(\sum_{i=1}^{K-1}p_i\beta_i\right)^2}{\sum_{i=1}^{K-1}p_i}.
\]
}\noindent
The equality holds if and only if $\beta_1=\cdots=\beta_{K-1}$. Therefore, $F(a)\geq0$.

(1.2) Assume that there exists some $j\neq K$ such that $q=p_j\geq0.5$. By plugging Eq. (\ref{eq:a1}) into $F(a)$, we have
{\small
\begin{align*}
&F(a)=2q\left(\sum_{i\neq j}^{K-1}p_i\beta_i^2+(1-q)\beta_j^2\right)
-\frac{q}{1-q}\left(\sum_{i\neq j}^{K-1}p_i\beta_i+(1-q)\beta_j\right)^2-\sum_{i=1}^{K-1}p_i\beta_i^2
\\
&+\left(\sum_{i=1}^{K-1}p_i\beta_i\right)^2
=(2q-1)\sum_{i\neq j}^{K-1}p_i\beta_i^2-\left(\frac{q}{1-q}-1\right)\left(\sum_{i\neq j}^{K-1}p_i\beta_i\right)^2
\geq 0.
\end{align*}
}\noindent
Again, the above inequality is obtained by the Cauchy-Schwartz inequality.

(1.3) Assume that $p_k<0.5$ for all $k=1,\ldots,K$ and hence $q=0.5$. Under this case, we can immediately obtain $F(a)=0$.

Combing (1.1), (1.2) and (1.3), we have shown that $F(a)\geq0$ under the first case, where $\gamma \geq 2q$.

\medskip
\noindent
(2) Consider the second case in Theorem \ref{thm:LUS}, where $1\leq\gamma<2q$.
Given data point $\x$, we consider the following three sub-cases.
\medskip

(2.1) Assume that $q=p_K\geq0.5$. By plugging Eq. (\ref{eq:a2}) into $F(a)$, we have
$\sum_{i=1}^Ka_ip_i=\frac{\gamma(1-q)}{\gamma-q}$, and
{\small
\begin{align*}
F(a)=(\gamma-1)\sum_{i=1}^{K-1}p_i\beta_i^2-\left(\frac{\gamma-q}{1-q}-1\right)\left(\sum_{i=1}^{K-1}p_i\beta_i\right)^2\geq0.
\end{align*}
}\noindent

(2.2) Assume that there exists some $j\neq K$ such that $q=p_j\geq0.5$. By plugging Eq. (\ref{eq:a2}) into $F(a)$, we have
{\small
\begin{align*}
&F(a)=\gamma\left(\sum_{i\neq j}^{K-1}p_i\beta_i^2+\frac{q(1-q)}{\gamma-q}\beta_j^2\right)
-\frac{\gamma-q}{1-q}\left(\sum_{i\neq j}^{K-1}p_i\beta_i+\frac{q(1-q)}{\gamma-q}\beta_j\right)^2
\\
&-\sum_{i=1}^{K-1}p_i\beta_i^2+\left(\sum_{i=1}^{K-1}p_i\beta_i\right)^2
=(\gamma-1)\sum_{i\neq j}^{K-1}p_i\beta_i^2-\left(\frac{\gamma-q}{1-q}-1\right)\left(\sum_{i\neq j}^{K-1}p_i\beta_i\right)^2
\geq 0.
\end{align*}
}\noindent

(2.3) Assume that $p_k<0.5$ for all $k=1,\ldots,K$ and hence $q=0.5$. Under this case, we can derive with the same way as above and obtain
{\small
\[
F(a) = (\gamma-1)\left[\sum_{i=1}^{K-1}p_i\beta_i^2-\left(\sum_{i=1}^{K-1}p_i\beta_i\right)^2\right]
\geq0 .
\]
}\noindent
Combing (2.1), (2.2) and (2.3), we have shown $F(a)\geq0$ under the second case $1 \leq \gamma < 2q$.
%\end{enumerate}

The first case and the second case together imply that Eq. (\ref{eq:PSD}) holds, which proves
$\gamma\mathbf{S}\succeq\mathbf{S}_{full}$. This leads to the desired variance bound in the theorem.

\medskip
We now study the expected sample size. 
%Again, we consider two cases below.
%\begin{enumerate}
%\item[(1)] 

\medskip
\noindent
(1) Consider the first case in Theorem \ref{thm:LUS}, where $\gamma\geq2q$: the conditional expectation of $a(\x,y)$ given $\x$ is
{\small
\[
\bar{a}(\x)=\E_{y|\x\sim\D}\ a(\x,y)= \sum_{k=1}^Ka_kp_k=\frac{4}{\gamma} q(1-q).
\]
}\noindent
Therefore, the point-wise conditional expectation of the acceptance probability satisfies $\bar{a}(\x)\leq\frac{1}{\gamma}$.

%\item[(2)] 

\noindent
(2) Consider the second case in Theorem \ref{thm:LUS}, where $1 \leq \gamma < 2q$:
the conditional expectation of $a(\x,y)$ given $\x$ is
{\small
\[
\bar{a}(\x)=\E_{y|\x\sim\D}\ a(\x,y)= \sum_{k=1}^Ka_kp_k=\frac{1}{\gamma} \frac{\gamma^2(1-q)}{\gamma-q}.
\]
}\noindent
Note that $0.5\leq q\leq1$ and $1\leq\gamma<2q$. We have
{\small
\[
\gamma \bar{a}(\x)-1 =\frac{\gamma^2(1-q)}{\gamma-q}-1
=\frac{(\gamma-1)(\gamma(1-q)-q)}{\gamma-q}<\frac{q(\gamma-1)(1-2q)}{\gamma-q}\leq0.
\]
}\noindent
Therefore, the point-wise conditional expectation of the acceptance probability satisfies $\bar{a}(\x)\leq\frac{1}{\gamma}$.
%\end{enumerate}

\medskip
Combing both (1) and (2), we know that the expected number of accepted examples is
{\small
\[
n_{Sub}=n \E_{\x \sim \D}\ \bar{a}(\x) \leq\frac{n}{\gamma}.
\]
}\noindent
This completes the proof of Theorem~\ref{thm:LUS}.
\end{proof}

\subsection{Proof of Corollary \ref{cor:pilot}}
\begin{proof}
The proof follows the structure in proof of Theorem \ref{thm:variance}. By the mean value theorem, we have
{\small
\begin{align*}
\sqrt{n}\left(\hat{\bT}_{LUS}-\bT^0 \right)=-\nabla_{\bT}^2\hat{R}_n(\hat{\bl},\bar{\bT})^{-1}\sqrt{n}\nabla_{\bT}\hat{R}_n(\hat{\bl},\bT^0) ,
\end{align*}
}\noindent
where $\bar{\bT}=t\bT^0+(1-t)\hat{\bT}_{LUS}$ for some $t\in[0,1]$. 
Since $\hat{\bT}_{LUS}\stackrel{p}{\rightarrow} \bT^0$, we have $\bar{\bT} \stackrel{p}{\rightarrow} \bT^0$. From the condition in the theorem, $\hat{\bl}\stackrel{p}{\rightarrow} \bT^0$. Therefore, by the continuous mapping theorem, $\nabla_{\bT}^2 \hat{R}_n(\hat{\bl},\bar{\bT})^{-1}$ converges to $\nabla_{\bT}^2 \hat{R}_n(\bT^0,\bT^0)^{-1}$.  Moreover, by Lemma~\ref{lem:1}, $\nabla_{\bT}^2 \hat{R}_n(\bT^0,\bT^0)^{-1}$ converges to $\nabla_{\bT}^2 R(\bT^0,\bT^0)^{-1}$ as $n\rightarrow\infty$. 
Finally, by the Slutsky's Theorem, the limiting distribution of $\sqrt{n}\left(\hat{\bT}_{LUS}-\bT^0\right)$ is given by
{\small
\[
-\nabla_{\bT}^2 R(\bT^0,\bT^0)^{-1}\sqrt{n}\nabla_{\bT}\hat{R}_n(\bT^0,\bT^0).
\]
}\noindent
Following the proof of Theorem \ref{thm:variance}, we have
{\small
\[
Var\left(\sqrt{n}\left(\hat{\bT}_{LUS}-\bT^0\right)\right)=[\E_{\x\sim\D}\bm{\nabla}(\bT^0)\mathbf{S}(\bT^0, \bT^0)\bm{\nabla}(\bT^0)^\top]^{-1},
\]
}\noindent
which is a constant that is independent of the pilot estimator $\hat{\bl}$.
\end{proof}

\subsection{Proof of Proposition \ref{prop:mis1}}

\begin{proof}
The goal is to show that $\bt^*$ is the maximizer of $R(\bt^*, \bt)$. When $K=2$, the label $y\in\{0,1\}$ and Eqs.~(\ref{eq:risk}) and (\ref{eq:LUSriskPop}) can be rewritten as
{\small
\begin{align}
L(\bt)&= \E_{\x,y \sim \D}\ \left[ \mathbb{I}(y=1) \cdot f(\x,\bt) - \log \left(1 + e^{f(\x,\bt)}\right) \right],
\label{eq:binaryL}
%\\
%&= \E_{\x\sim\D}\ p(\x,1)f(\x,\bt)-\log \left(1 + e^{f(\x,\bt)}\right).
%\nonumber
\end{align}
and 
\begin{align}
R(\bl,\bt)&= \E_{\x,y,z \sim \A}\ z\left[ \mathbb{I}(y=1) \cdot g(\x,\bt) - \log \left(1 + e^{g(\x,\bt)}\right) \right]
\label{eq:binaryrisk}
\\
&= \E_{\x\sim\D}\ a_{\bl}(\x,1)p(\x,1)
\left[
g(\x,\bt)-\log \left(1 + e^{g(\x,\bt)}\right)
\right]
\nonumber
\\
&\ \ \ \ \ \ \ \ \ \ \ \ 
-a_{\bl}(\x,0)p(\x,0)\log \left(1 + e^{g(\x,\bt)}\right),
\nonumber
\end{align}
}\noindent
where we write $a_{\bl} (\cdot,\cdot)$ to indicate that the acceptance probability is computed based on the pilot estimator $\bl$.
The first order and second order derivatives of $R(\bl,\bt)$ with respect to $\bt$ are
{\small
\[
\nabla_{\bt}R(\bl,\bt)=\E_{\x}\left[a_{\bl}(\x,1)p(\x,1)\frac{1}{1+e^{g(\x,\bt)}}-a_{\bl}(\x,0)p(\x,0)\frac{e^{g(\x,\bt)}}{1+e^{g(\x,\bt)}}\right]\nabla_{\bt}f(\x,\bt)
\] 
}
and 
{\small
\[
\nabla_{\bt}^2R(\bl,\bt)=-\E_{\x}\left[a_{\bl}(\x,1)p(\x,1)+a_{\bl}(\x,0)p(\x,0)\right]\frac{e^{g(\x,\bt)}}{(1+e^{g(\x,\bt)})^2}\nabla_{\bt}f(\x,\bt)\nabla_{\bt}f(\x,\bt)^{\top},
\]
}
respectively. 

In this case, the misspecified logistic model is defined as
{\small
\begin{align}
\label{proof lem:identity1}
f(\x,\bt^*)=\log\frac{p^*(\x,1)}{p^*(\x,0)}.
\end{align}
}\noindent
Suppose $q(\x)=p^*(\x,1)=\frac{e^{f(\x,\bt^*)}}{1+e^{f(\x,\bt^*)}}\geq0.5$. For the choice of $\gamma$ either in Eq.~(\ref{eq:a1}) or Eq.~(\ref{eq:a2}), we have
{\small
\begin{align*}
&\nabla_{\bt}R(\bt^*,\bt^*)
\\
&=\E_{\x\sim\D}
\left[
a(\x,1)p(\x,1)\left(1-\frac{e^{g(\x,\bt^*)}}{1+e^{g(\x,\bt^*)}}\right)
-a(\x,0)p(\x,0)\frac{e^{g(\x,\bt^*)}}{1+e^{g(\x,\bt^*)}}
\right]\nabla_{\bt}f (\x,\bt^*)
\nonumber
\\
&=\E_{\x\sim\D}
\left[
a(\x,1)p(\x,1)\left(1-\frac{a(\x,1)e^{f(\x,\bt^*)}}{a(\x,0)+a(\x,1)e^{f(\x,\bt^*)}}\right)
\right.
\\
&\ \ \ 
\left.
-a(\x,0)p(\x,0)\frac{a(\x,1)e^{f(\x,\bt^*)}}{a(\x,0)+a(\x,1)e^{f(\x,\bt^*)}}
\right]\nabla_{\bt}f (\x,\bt^*)
\nonumber
\\
&=\E_{\x\sim\D}
\frac{1}{\gamma}[a(\x,1)p(\x,1)-a(\x,0)p(\x,0)]\nabla_{\bt}f (\x,\bt^*)
\nonumber
\\
&=
\E_{\x\sim\D}
\frac{1}{\gamma}\left[(1-p^*(\x,1))p(\x,1)-p^*(\x,1)(1-p(\x,1))\right]\nabla_{\bt}f (\x,\bt^*)
\nonumber
\\
&=
\frac{1}{\gamma}\nabla_{\bt}L(\bt^*),
\nonumber
\end{align*}
}\noindent
where the third equality holds by plugging in Eq.~(\ref{proof lem:identity1}) and the definition of the acceptance probability in either Eq. (\ref{eq:a1}) or Eq. (\ref{eq:a2}). Since $\bt^*$ is the best estimator of the original MLE problem in Eq.~(\ref{eq:binaryL}), $\nabla_{\bt}R(\bt^*,\bt^*)=\frac{1}{\gamma}\nabla_{\bt}L(\bt^*)=\bm{0}$. Therefore, $\bt^*$ is the maximizer of $R(\bt^*,\bt)$.
The case of $q(\x)=p^*(\x,0)\geq0.5$ can be derived similarly because of the symmetry. Given $\bl$, $R(\bl,\bt)$ is strictly concave on $\bt$ because of the expectation over the population. With the assumption that $\P_\D(f(\x,\bt)\neq f(\x,\bt^*))>0$ for all $\bt\neq\bt^*$, we have $R(\bt^*,\bt)<R(\bt^*,\bt^*)$ for any $\bt\neq\bt^*$. Therefore, $\bt^*$ is the unique maximizer of $R(\bt^*,\bt)$.

%Now, since the parameter space is compact, we have $\P_{\D}(f(\x,\bt) \neq f(\x,\bt^*)) >0$, $\forall\bt\neq\bt^*$.  Therefore, for any $\bt\neq\bt^*$, we have $R(\bt) < R(\bt^*)$.  
It follows that given any $\epsilon'>0$, there exists $\epsilon>0$ so that $R(\bt^*,\bt) \geq R(\bt^*,\bt^*)-2 \epsilon$ implies that $\|\bt-\bt^*\|<\epsilon'$. Now according to Lemma \ref{lem:1}, given any $\delta >0$, when $n \to \infty$, with probability larger than $1-\delta$ we have
{\small
\[
R(\bt^*,\hat{\bt}_{LUS}) \geq \hat{R}_n(\bt^*,\hat{\bt}_{LUS}) - \epsilon
\geq  \hat{R}_n(\bt^*,\bt^*) - \epsilon
\geq  R(\bt^*,\bt^*) - 2 \epsilon .
\]
}\noindent
This implies that $\|\hat{\bt}_{LUS}-\bt^*\|<\epsilon'$. So, the empirical estimator $\hat{\bt}_{LUS}$ converges to $\bt^*$ as $n\rightarrow\infty$.
\end{proof}

\subsection{Proof of Proposition \ref{prop:mis2}}
\begin{proof}
%We first show $\hat{\bt}_{LUS}\rightarrow\bt^*$. Define the convex function $G(g)=\log(1+e^g)$. The Bregman divergence of the convex function $G(g)$ is 
%{\small
%\begin{align*}
%\Delta(g, g')
%=G(g)-G(g')-\nabla G(g')^{\top}(g-g')
%=\log\frac{1+e^g}{1+e^{g'}}+\frac{e^{g'}}{1+e^{g'}} (g'-g).
%\end{align*}
%}\noindent
%Since $G(\cdot)$ is convex, we have $\Delta(g, g')\geq0$ and $\Delta(g, g')=0$ only when $g=g'$. 
%
%Given any $h(\x)$, if we define
%{\small
%\[
%\tilde{R}(h(\x))=\E_{\x,c,z\sim\A}\ z\left[\mathbb{I}(c=1)h(\x)-\log(1+e^{h(\x)})\right],
%\]
%}\noindent
%then $h(\x)=g(\x,\bt^*)$ is the global optimizer of $\tilde{R}(\cdot)$ and $R(\bt)=\tilde{R}(g(\x,\bt^*))$. It implies that,
%{\small
%\begin{align*}
%&R(\bt^*)-R(\bt)=
%\E_{\x,c\sim\D}\ a(\x,c)\Delta(g(\x,\bt), g(\x,\bt^*)).
%\end{align*}
%}\noindent
%
%The assumption of the theorem implies that the parameter space is compact and we have $\P_{\D}(f(\x,\bt)\neq f(\x,\bt^*)) >0$, $\forall\bt\neq\bt^*$. Therefore, for any $\bt\neq\bt^*$, we have $R(\bt)<R(\bt^*)$. 
%
%It follows that given any $\epsilon'>0$, there exists $\epsilon>0$ so that $R(\bt) \geq R(\bt^*)-2 \epsilon$ implies that $\|\bt-\bt^*\|<\epsilon'$. Now according to Lemma \ref{lem:1}, given any $\delta >0$, when $n \to \infty$, with probability larger than $1-\delta$ we have
%{\small
%\[
%R(\hat{\bt}_{LUS}) \geq \hat{R}_n(\hat{\bt}_{LUS}) - \epsilon
%\geq  \hat{R}_n(\bt^*) - \epsilon
%\geq  R(\bt^*) - 2 \epsilon .
%\]
%}\noindent
%This implies that $\|\hat{\bt}_{LUS}-\bt^*\|<\epsilon'$. 
By following the proofs of Proposition 3 and Theorem 5 in \cite{fithian2014local}, we can show that $\hat{\bt}_{LUS}\xrightarrow{p}\bt^*$ when $\hat{\bl}\xrightarrow{p}\bt^*$. In this proof, we focus on the analysis of the asymptotic variance. Let 
{\small
\[
\bar{\bt}(\hat{\bl})=\arg\max_{\bt}R(\hat{\bl},\bt).
\]
}
By the mean value theorem, we have
{\small
\[
\nabla_{\bt}\hat{R}(\hat{\bl}, \hat{\bt}_{LUS})=\nabla_{\bt}\hat{R}(\hat{\bl}, \bar{\bt}(\hat{\bl}))+
\nabla_{\bt}^2\hat{R}(\hat{\bl}, \bt_n')(\hat{\bt}_{LUS}-\bar{\bt}(\hat{\bl})),
\]
}\noindent
where $\bt_n'$ is some convex combination of $\hat{\bt}_{LUS}$ and $\bar{\bt}(\hat{\bl})$. Rearranging the equation, we obtain
{\small
\[
\sqrt{n}(\hat{\bt}_{LUS}-\bar{\bt}(\hat{\bl}))=-\nabla_{\bt}^2\hat{R}(\hat{\bl}, \bt_n')^{-1}
\cdot\sqrt{n}\nabla_{\bt}\hat{R}(\hat{\bl}, \bar{\bt}(\hat{\bl})).
\]
}\noindent
From the condition in the proposition, $\hat{\bl}\xrightarrow{p}\bt^*$ and $\bar{\bt}(\hat{\bl})\xrightarrow{p}\bar{\bt}(\bt^*)=\bt^*$ from Proposition \ref{prop:mis1}. Moreover, $\hat{\bt}_{LUS}\xrightarrow{p}\bt^*$ and hence $\bt_n'\xrightarrow{p}\bt^*$. By the continuous mapping theorem, $\nabla_{\bt}^2 \hat{R}_n(\hat{\bl},\bar{\bt})^{-1}$ converges to $\nabla_{\bt}^2 \hat{R}_n(\bt^*,\bt^*)^{-1}$ which will eventually converge to $\nabla_{\bt}^2 R(\bt^*,\bt^*)^{-1}$ as $n\rightarrow\infty$. By the Slutsky's Theorem, the limiting distribution of $\sqrt{n}(\hat{\bt}_{LUS}-\bt^*)$ is
{\small
\begin{align}
\sqrt{n}(\hat{\bt}_{LUS}-\bar{\bt}(\hat{\bl}))\xrightarrow{d}\mathcal{N}(\bm{0},
\H(\bt^*,\bt^*)^{-1}J(\bt^*,\bt^*)\H(\bt^*,\bt^*)^{-1}).
\label{eq:bridge}
\end{align}
}\noindent
%Following the techniques used in W. Fithian and T. Hastie's paper, we can show the first term tends in probability to $\nabla_{\bt}^2 R(\bt^*,\bt^*)^{-1}$ and the second term tends in distribution to $\mathcal{N}(\bm{0}, J(\bt^*,\bt^*))$, then by the Slutsky's theorem we obtain the result.

Define the cross partial derivative
{\small
\[
C(\bl,\bt)=\frac{\partial^2R(\bl,\bt)}{\partial\bt\partial\bl}.
\]
}\noindent
By viewing $\bar{\bt}(\bl)$ as a function of $\bl$, we apply the Implicit Function Theorem to the fact $\nabla_{\bt}R(\bl,\bar{\bt}(\bl))=\mathbf{0}$ such that the implicit derivative $\nabla_{\bl}\bar{\bt}(\bl)$ can be written as
{\small
\[
\nabla_{\bl}\bar{\bt}(\bl)=-\H(\bl,\bar{\bt}(\bl))^{-1}C(\bl,\bar{\bt}(\bl)).
\]
}\noindent
Then, by using the Taylor expansion and noting that $\bar{\bt}(\bt^*)=\bt^*$, we have
{\small
\begin{align*}
\bar{\bt}(\bl)
&=\bar{\bt}(\bt^*)+\nabla_{\bl}\bar{\bt}(\bl)|_{\bl=\bt^*}(\bl-\bt^*)+o(\|\bl-\bt^*\|)
\\
&=\bt^*-\H(\bt^*,\bt^*)^{-1}C(\bt^*,\bt^*)(\bl-\bt^*)+o(\|\bl-\bt^*\|),
\end{align*}
}\noindent
for some norm $\|\cdot\|$. Rearranging the terms and as $\hat{\bl}\xrightarrow{p}\bt^*$ we obtain
{\small
\begin{align*}
Var(\sqrt{n}(\bar{\bt}(\hat{\bl})-\bt^*))
&=\H(\bt^*,\bt^*)^{-1}C(\bt^*,\bt^*)\mathcal{V}_{\lambda}C(\bt^*,\bt^*)^{\top}\H(\bt^*,\bt^*)^{-1}.
\end{align*}
}\noindent
We reach the desired result by combining the above with Eq.~(\ref{eq:bridge}).
\end{proof}

\subsection{Proof of Proposition \ref{prop:mis3}}

\begin{proof}
%By absorbing $\x$, $\bl$ and $\bT$, we use the notations $g_k$, $f_k$, $a_k$, $p_k$ and $p_k^*$ to denote $g(\x,\bt_k)$, $f(\x,\bt_k)$, $a_{\bl}(\x,k)$, $p(\x,k)$ and $p(\x,\bt_k^*)$.
The population version with respect to $\A$ is
{\small
\begin{align*}
R(\bl,\bT)&=\E_{\x,y,z\sim\A}\ z\left[\sum_{k=1}^{K-1}\mathbb{I}(y=k)g_k-\log\left(1+\sum_{k=1}^{K-1}e^{g_k}\right) \right]
\\
&=\E_{\x\sim\A}\sum_{k=1}^{K-1}a_k p_k g_k-\log\left(1+\sum_{k=1}^{K-1}e^{g_k}\right)\sum_{k=1}^{K}a_k p_k.
\end{align*}
}
We measure the gradient with respect to $\bt_k$ by letting $\bT=\bT^*$ as
{\small
\begin{align*}
\nabla_{\bt_k}R(\bl,\bT^*)&=\E_{\x}\left[a_k p_k-\frac{e^{g_k^*}}{1+\sum_{k=1}^{K-1}e^{g_k^*}}\sum_{k=1}^{K}a_k p_k\right]\nabla_{\bt_k}f_k^*
\\
&=\E_{\x}\left[a_k p_k-\frac{a_k p_k^*}{\sum_{k=1}^K a_k p_k^*}\sum_{k=1}^{K}a_k p_k
\right]\nabla_{\bt_k}f_k^*.
\end{align*}
}\noindent
Now, given $\x$, suppose there exists a majority class $i$ such that $p_i^*\geq0.5$. Then, if the acceptance probability $a_k$ is computed from a pilot that is $\bT^*$, for choice of $\gamma$ in Eq.~(\ref{eq:a1}), we have
{\small
\[
\nabla_{\bt_i}R(\bT^*,\bT^*)=\E_{\x}\ \frac{1}{\gamma}(p_i-p_i^*)\nabla_{\bt_i}f_i^*=\frac{1}{\gamma}\E_{\x}\nabla_{\bt_i}L(\bT^*)\equiv \bm{0},
\]
}\noindent
where $L(\bT)$ is the likelihood defined over the original population $\D$. Since $\bT^*$ is the best estimator of $\arg\max_{\bT}L(\bT)$ under misspecification, the above quantity is $\bm{0}$.
For choice of $\gamma$ in Eq.~(\ref{eq:a2}), we can have the same result. That is, the partial derivative of $R(\bT^*,\bT^*)$ with respect to the majority class $i$ is always $\bm{0}$. Now, for class $j\neq i$, with choice of Eq.~(\ref{eq:a1}), we have
{\small
\[
\nabla_{\bt_j}R(\bT^*,\bT^*)=\E_{\x}\ \frac{2}{\gamma}
\left\{p_i^*p_j-\big[p_i(1-p_i^*)+p_i^*(1-p_i)\big]\frac{p_j^*}{2(1-p_i^*)}
\right\}\nabla_{\bt_j}f_j^*.
\]
}\noindent
Generally, the partial derivative with respect to any non-majority class $j$ with $\bt_j^*$ is not $\bm{0}$, which implies that the overall $\bT^*$ is not the maximizer of $R(\bT^*,\bT)$ even if we choose $\bl=\bT^*$. So, the LUS estimator is not consistent to $\bT^*$ for $K>2$. Similar results can be derived for the choice of Eq.~(\ref{eq:a2}).
\end{proof}

\section{Bias under Model Misspecification when $K>2$.}\label{appendix:bias}
%In the following proposition, we perform an analysis to discuss the effect of model misspecification on both the bias and variance of the LUS estimator. To this end, we denote the probability estimate with $\bT^*$ as $\bm{p}^*=(p^*(\x,1), \ldots, p^*(\x,K))^\top$ and denote the true probability distribution as $\bm{p}=(p(\x,1), \ldots, p(\x,K))^\top$. In addition, let $\bm{a}^*=(a^*(\x,1),\cdots,a^*(\x,K))^\top$ be the acceptance probability by using $\bm{\lambda}=\bT^*$ and let $\tilde{\bT}^*$ be the maximizer of $R(\bT)$ in Eq. \eqref{eq:LUSriskPop}.  

%\textcolor{blue}{
%Let $\tilde{\bT}^*$
%}
We have the following propositions quantifying the bias of the LUS estimator under model misspecification.

\begin{prop}\label{prop:mis4}
Consider a misspecified model with $K>2$. Let $\tilde{\bT}^*=\arg\max_{\bT}R(\bl,\bT)$. With the following assumptions as considered similarly in the case of correctly specified model:
\begin{itemize}
	\item the pilot estimator $\bm{\lambda}$ is set to $\bT^*$;
	\item $\|\nabla_{\bt_k}f(\x,\bt_k)\|$ and $\|\nabla_{\bt_k}^2f(\x,\bt_k)\|$ are bounded by constant $M$ for some norm $\|\cdot\|$ and $k=1,\cdots,K-1$;  
%	\item the Hessian $\nabla^2 R(\bT^*,\bT)$ is negative-definite given any $\bT$;
	\item there exists positive scalars $\omega>0$ and $\epsilon>0$ so that $\|\bm{p}^*-\bm{p}\|\leq\omega$ and $\|\tilde{\bT}^*-\bT^*\|\leq\epsilon$.
\end{itemize}
%we have $\tilde{\bT}^*\neq\bT^*$.
Then, if $\omega\rightarrow0$ is sufficiently small, then $\epsilon\rightarrow0$ as well.
\end{prop}
\begin{proof}
Let $\tilde{\bT}^*=\bT^*+\epsilon\bm{\varphi}$ with normalized non-zero vector $\bm{\varphi}$. Then, by the second order mean value theorem, there exists a $t \in [0,1]$ such that
{\small
\[
	R(\bT^*,\tilde{\bT}^*)=R(\bT^*,\bT^*)+\epsilon\bm{\varphi}^\top\nabla_{\bT} R(\bT^*,\bT^*)
	+\frac{\epsilon^2}{2}\bm{\varphi}^\top\nabla_{\bT}^2 R(\bT^*, \bT^* + t(\tilde{\bT}^*-\bT^*))\bm{\varphi}.
\]
}\noindent
%Using the fact that $\tilde{\bT}^*$ is the maximizer of $R(\bT)$ and that $\bT^*$ is not the maximizer of $R(\bT)$, we have $R(\tilde{\bT}^*) - R(\bT^*)> 0$.  Moreover, $\nabla^2 R(\bT)$ is negative-definite.  Thus, this yields the inequality
%\[
%- \bm{\varphi}^\top\nabla R(\bT^*)
%	< \frac{\epsilon}{2}\bm{\varphi}^\top\nabla^2 R( t\bT^* + t(\tilde{\bT}^*-\bT^*))\bm{\varphi} < 0.
%\]
%It can be shown that $\|\nabla R(\bT^*)\|\leq 3\omega M$, then as long as $\omega \rightarrow 0$, we have $\nabla R(\bT^*)\rightarrow \mathbf{0}$.
%Thus, this implies that $\epsilon \rightarrow 0$.  
%
%Denote $\tilde{\bT}^*=\bT^*+\epsilon\bm{\varphi}$ with $\epsilon>0$ and normalized non-zero vector $\bm{\varphi}$. Then, by the second order mean value theorem, there exists a $t\in [0,1]$ such that 
%\begin{align*}
%	R(\tilde{\bT}^*)=R(\bT^*)+\epsilon\bm{\varphi}^\top\nabla R(\bT^*)
%	+\frac{\epsilon^2}{2}\bm{\varphi}^\top\nabla^2 R(\bT^* + t(\tilde{\bT}^*-\bT^*))\bm{\varphi}.
%\end{align*}
Let $\bar{\bT}=\bT^* + t(\tilde{\bT}^*-\bT^*)$ and use the notations $\bar{\bm{p}}$ and $\bar{\nabla}_j$ to represent the probability estimate with $\bar{\bT}$ and $\nabla_{\bt_j}f(\x,\bar{\bt}_j)$. The elements in the Hessian matrix $\nabla_{\bT}^2 R(\bT^*,\bar{\bT})$ take the form
{\small
\begin{align*}
\nabla_{jj}^2R(\bT^*,\bar{\bT})
&=-\E_{\x,y\sim\D}\ 
a_y^*\cdot\frac{a_j^*\bar{p}_j\sum_{k\neq j}^Ka_k^*\bar{p}_k}{(\sum_{k=1}^Ka_k^*\bar{p}_k)^2}\cdot
\bar{\nabla}_j{\bar{\nabla}_j}^\top,
\\
&=-\E_{\x}
\left(\sum_{k=1}^{K}a_k^*p_k\right)\cdot\frac{a_j^*\bar{p}_j\sum_{k\neq j}^Ka_k^*\bar{p}_k}{(\sum_{k=1}^Ka_k^*\bar{p}_k)^2}\cdot
\bar{\nabla}_j{\bar{\nabla}_j}^\top,
\end{align*}
\begin{align*}
\nabla_{ij}^2R(\bT^*, \bar{\bT})
&=\E_{\x,y\sim\D}\ 
a_y^*\cdot\frac{a_i^*\bar{p}_ia_j^*\bar{p}_j}{(\sum_{k=1}^Ka_k^*\bar{p}_k)^2}\cdot
\bar{\nabla}_i{\bar{\nabla}_j}^\top
\\
&=\E_{\x}
\left(\sum_{k=1}^{K}a_k^*p_k\right)\cdot\frac{a_i^*\bar{p}_ia_j^*\bar{p}_j}{(\sum_{k=1}^Ka_k^*\bar{p}_k)^2}\cdot
\bar{\nabla}_i{\bar{\nabla}_j}^\top,\ \ \ \text{for}\ \ i\neq j.
\end{align*}
}\noindent
Therefore, we have
{\small
\begin{align*}
&\nabla_{\bT}^2 R(\bT^*,\bar{\bT})=-\E_{\x}\ \bm{\nabla}(\bar{\bT})\bm{S}'(\bT^*,\bar{\bT})\bm{\nabla}(\bar{\bT})^\top,
\end{align*}
}\noindent
where $\bm{S}'(\bT^*,\bar{\bT})=\frac{\sum_{k=1}^Ka_k^*p_k}{\sum_{k=1}^Ka_k^*\bar{p}_k}\bm{S}(\bT^*,\bar{\bT})\succ0$. 
Since $R(\bT^*,\bT)$ is strictly convex on $\bT$, $\nabla_{\bT}^2R(\bT^*,\bT)$ is negative-definite, and we have
{\small
\[
\frac{\epsilon^2}{2}\bm{\varphi}^\top\nabla_{\bT}^2 R(\bT^*, \bT^* + t(\tilde{\bT}^*-\bT^*))\bm{\varphi}<0.
\]
}

Next, we show that if $\omega\rightarrow0$, $\nabla_{\bT} R(\bT^*, \bT^*)\rightarrow\bm{0}$. The gradient with respect to the $i$-th class is
{\small
\begin{align*}
\nabla_{\bt_i} R(\bT^*, \bT^*)=
\left[a_ip_i-\sum_{k=1}^Ka_kp_k\cdot\frac{a_ip_i^*}{\sum_{k=1}^Ka_kp_k^*}\right]\nabla_i^*.
\end{align*}
}\noindent

From Proposition \ref{prop:mis3}, we know that if there exists $j$ such that $p_j^*\geq0.5$ then
 $\nabla_{\bt_j} R(\bT^*, \bT^*)=\bm{0}$ and $\nabla_{\bt_i} R(\bT^*, \bT^*)\neq\bm{0}$ for $i\neq j$.

Since $|p_j^*-p_j|\leq \omega$, by writing $p_j^*=p_j+\omega\delta_j$ and $p_i^*=p_i+\omega\delta_i$ for some $|\delta_i|\leq1$ and $|\delta_j|\leq1$, we have
{\small
\begin{align*}
\nabla_{\bt_i} R(\bT^*, \bT^*)=\E_{\x}\ \omega\delta_j\cdot\left[2p_i-(1-2p_i)\frac{p_i+\omega\delta_i}{1-p_j-\omega\delta_j}\right]\nabla_i^*,
\end{align*}
and when $\omega\rightarrow0$,
\begin{align}
\|\nabla_{\bt_i} R(\bT^*,\bT^*)\|\leq3\omega M\rightarrow0.
\label{eq:gradOdelta1}
\end{align}
}\noindent
Similar conclusion can be obtained for the case of $\gamma<2p_j^*$.

Since $\tilde{\bT}^*$ is the maximizer of $R(\bT^*,\bT)$, we have $R(\bT^*,\bT^*)<R(\bT^*,\tilde{\bT}^*)$. Now, combining all of the above, we obtain
{\small
\[
-\bm{\varphi}^\top\nabla_{\bT} R(\bT^*,\bT^*)
	< \frac{\epsilon}{2}\bm{\varphi}^\top\nabla_{\bT}^2 R(\bT^*,\bT^*  + t(\tilde{\bT}^*-\bT^*))\bm{\varphi} <0.
\]
}\noindent
Thus, when $\omega\rightarrow0$, $\nabla_{\bT} R(\bT^*,\bT^*)\rightarrow\bm{0}$, and we have $\epsilon\rightarrow 0$.
\end{proof}

\begin{prop}\label{prop:mis5}
Assume the same conditions in Proposition~\ref{prop:mis4}, with norm $\|\cdot\|=\|\cdot\|_\infty$. As $n\rightarrow\infty$, we have
{\small
\[
\|\hat{\bT}_{LUS}-\bT^*\|_{\infty}\leq
3\omega MKd\left[\|\H^{-1}\|_{\infty}+\delta\|\H^{-1}(\mathbf{I}+\delta\H^{-1}\Delta)^{-1}\H^{-1}\|_{\infty}\right],
\]
}\noindent
where
$\delta=(\omega+2\epsilon K)M^2$ and $\H=\E_{\x\sim\D}\bm{\nabla}(\bT^*)\mathbf{S}(\bT^*,\bT^*)\bm{\nabla}(\bT^*)^\top$ are constants, and $\Delta=\E_{\x\sim\D}\bm{\nabla}(\bT^*)\mathbf{E}\bm{\nabla}(\bT^*)^\top$ with $\mathbf{E}\in\mathbb{R}^{(K-1)\times(K-1)}$ satisfying $\|\mathbf{E}\|_{\infty}\leq1$. 
%If $\omega$ is sufficiently small, we have $\bm{\mu}\rightarrow\bm{0}$ and $\mathcal{V}_{LUS}^{miss}\preceq\gamma\mathcal{V}_{full}^{miss}$.
\end{prop}

\begin{proof}
By the mean value theorem, there exists some linear combination $\bT'=t\hat{\bT}_{LUS}+(1-t)\bT^*$ for $t\in[0,1]$ such that
{\small
\begin{align*}
\hat{\bT}_{LUS}-\bT^*
&=\nabla_{\bT}^2 \hat{R}_n(\bT^*,\bT')^{-1}[\nabla_{\bT}\hat{R}_n(\bT^*,\hat{\bT}_{LUS})-\nabla_{\bT}\hat{R}_n(\bT^*,\bT^*)]
\\
&=-\nabla_{\bT}^2 \hat{R}_n(\bT^*,\bT')^{-1}\nabla_{\bT}\hat{R}_n(\bT^*,\bT^*),
\end{align*}
}\noindent
where the second equality holds since $\hat{\bT}_{LUS}$ is the maximizer of $\hat{R}_n(\bT^*,\bT)$. Suppose there exists a $\bar{\bT}$ that $\bT'\xrightarrow{p}\bar{\bT}$ as $n\rightarrow\infty$. According to Lemma~\ref{lem:1} and the continuous mapping theorem, $\nabla_{\bT}\hat{R}_n(\bT^*,\bT^*)\xrightarrow{p}\nabla_{\bT}R(\bT^*,\bT^*)$ and $\nabla_{\bT}^2\hat{R}_n(\bT^*,\bT')\xrightarrow{p}\nabla_{\bT}^2R(\bT^*,\bar{\bT})$.
Then, by using Eq. (\ref{eq:gradOdelta1}),
{\small
\[
\|\hat{\bT}_{LUS}-\bT^*\|_{\infty}
\leq 3\omega MKd\|\nabla_{\bT}^2R(\bT^*,\bar{\bT})^{-1}\|_{\infty}.
\]
}\noindent
%Moreover, the variance is
%{\small
%\[
%\mathcal{V}_{LUS}^{miss}=
%\nabla^2R(\bT^*,\bar{\bT}^0)^{-1}Var\left(\sqrt{n}\nabla\hat{R}_n(\bT^*)\right)\nabla^2R(\bT^*,\bar{\bT}^0)^{-1}.
%\]
%}\noindent
%Let $U(\bT^*)=Var\left(\sqrt{n}\nabla\hat{R}_n(\bT^*,\bT^*)\right)$. 
%Then, we can obtain the following bounds
%{\small
%\begin{align*}
%&\|-\nabla^2R(\bT^*)-\bm{V}\|_{\infty}\leq \omega M^2,
%\\
%&\|\nabla^2R(\bar{\bT}^0)-\nabla^2R(\bT^*)\|_{\infty}\leq 2\epsilon M^2 (1+M),
%\\
%&\|U(\bT^*)-\bm{V}\|_{\infty}\leq \omega M^2.
%\end{align*}
%}\noindent
%Let $U(\bT^*)=\bm{V}+\sigma_1\tilde{\bm{\Delta}}_1$ and $-\nabla^2R(\bar{\bT}^0)=\bm{V}+\sigma_2\tilde{\bm{\Delta}}_2$.
%By using similar technique as in the proof of Corollary \ref{cor:var} and dropping small terms of order $O(\sigma_1\sigma_2^2)$, $O(\sigma_1\sigma_2)$ and $O(\sigma_2^2)$ for sufficiently small $\omega$ and $\epsilon$, we can obtain the desired results.
We now derive the following upper bound:
{\small
\begin{align}
&\|\nabla_{\bT}^2R(\bT^*,\bar{\bT})-\nabla_{\bT}^2R(\bT^*,\bT^*)\|_{\infty}\leq 2\epsilon M^2K,
\label{eq:b2}
\end{align}
}\noindent

With some abuse of notation, we will denote $\nabla_{\bt_i}f(\x,\bt_i^*)$ and $\nabla_{\bt_i}f(\x,\bar{\bt}_i)$ as $\nabla_i^*$ and $\bar{\nabla}_i$, respectively, and use the notations $a_k^*$ and $p_k^*$ to represent the acceptance probability and the probability estimate computed from $\bT^*$.
The elements in the matrix $\nabla_{\bT}^2R(\bT^*,\bT^*)$ take the form
{\small
\begin{align*}
&\nabla_{jj}^2R(\bT^*,\bT^*)
=-\E_{\x}
\left(\sum_{k=1}^{K}a_k^*p_k\right)\cdot\frac{a_j^*p_j^*\sum_{k\neq j}^Ka_k^*p_k^*}{(\sum_{k=1}^Ka_k^*p_k^*)^2}\cdot
\nabla_j^*{\nabla_j^*}^\top,
\\
&\nabla_{ij}^2R(\bT^*,\bT^*)
=-\E_{\x}
\left(\sum_{k=1}^{K}a_k^*p_k\right)\cdot\frac{a_i^*p_i^*a_j^*p_j^*}{(\sum_{k=1}^Ka_k^*p_k^*)^2}\cdot
\nabla_i^*{\nabla_j^*}^\top, \ \ \ \text{for}\ \ i\neq j.
\end{align*}
}\noindent

The elements of $\nabla_{\bT}^2R(\bT^*,\bar{\bT})$ take the form of
{\small
\begin{align*}
&\nabla_{jj}^2R(\bT^*,\bar{\bT})=-\E_{\x}
\left(\sum_{k=1}^{K}a_k^*p_k\right)\cdot\frac{a_j^*\bar{p}_j\sum_{k\neq j}^Ka_k^*\bar{p}_k}{(\sum_{k=1}^Ka_k^*\bar{p}_k)^2}\cdot
\bar{\nabla}_j\bar{\nabla}_j^{\top},
\\
&\nabla_{ij}^2R(\bT^*,\bar{\bT})=-\E_{\x}
\left(\sum_{k=1}^{K}a_k^*p_k\right)\cdot\frac{a_i^*\bar{p}_i a_j^*\bar{p}_j}{(\sum_{k=1}^Ka_k^*\bar{p}_k)^2}\cdot
\bar{\nabla}_i\bar{\nabla}_j^{\top}, \ \ \ \text{for} \ \ i\neq j.
\end{align*}
}\noindent
Let $F_j(\bar{\bT})=\frac{a_j^*\bar{p}_j}{\sum_{k=1}^Ka_k^*\bar{p}_k}\bar{\nabla}_j$ and
$G_j(\bar{\bT})=\frac{\sum_{k\neq j}^Ka_k^*\bar{p}_k}{\sum_{k=1}^Ka_k^*\bar{p}_k}\bar{\nabla}_j$ for any $j$. We have
{\small
\begin{align*}
\|\nabla_{\bt_j}F_j(\bar{\bT})\|_{\infty}
&=\left\|
\frac{a_j^*\sum_{k\neq j}^Ka_k^*\bar{p}_k}{(\sum_{k=1}^Ka_k^*\bar{p}_k)^2}
\nabla_{\bt_j}\bar{p}_j
+\frac{a_j^*\bar{p}_j}{\sum_{k=1}^Ka_k^*\bar{p}_k}\bar{\nabla}_j^2
\right\|_{\infty}
\leq2M,
\\
\|\nabla_{\bt_j}G_j(\bar{\bT})\|_{\infty}
&=\left\|
-\frac{a_j^*\sum_{k\neq j}^Ka_k^*\bar{p}_k}{(\sum_{k=1}^Ka_k^*\bar{p}_k)^2}
\nabla_{\bt_j}\bar{p}_j
+\frac{\sum_{k\neq j}^Ka_k^*\bar{p}_k}{\sum_{k=1}^Ka_k^*\bar{p}_k}\bar{\nabla}_j^2
\right\|_{\infty}
\leq2M,
\end{align*}
}\noindent
for any $\bar{\bT}$. 
Given $\|\tilde{\bT}^*-\bT^*\|_{\infty}\leq\epsilon$, we have $\|\bar{\bT}-\bT^*\|_{\infty}\leq\epsilon$.
Then, by the mean value theorem we can bound
$\|F_j(\bar{\bT})-F_j(\bT^*)\|_{\infty}\leq2\epsilon MK$ and $\|G_j(\bar{\bt})-G_j(\bt^*)\|_{\infty}\leq2\epsilon MK$ for any $j$. Then, we have the following
{\small
\begin{align*}
&\|\nabla_{jj}^2R(\bT^*,\bar{\bT})-\nabla_{jj}^2R(\bT^*,\bT^*)\|_{\infty}
\\
=&\ 
\|\E_{\x}\ \bm{a}^{*\top}\bm{p} [F_j(\bar{\bT})G_j(\bar{\bT})^{\top}-F_j(\bT^*)G_j(\bT^*)^{\top}]\|_{\infty}
\\
=&\ 
\|\E_{\x}\ \bm{a}^{*\top}\bm{p}[F_j(\bar{\bT})G_j(\bar{\bT})^{\top}-F_j(\bT^*)G_j(\bar{\bT})^{\top}+F_j(\bT^*)G_j(\bar{\bT})^{\top}-F_j(\bT^*)G_j(\bT^*)^{\top}]\|_{\infty}
\\
\leq &\ 
\|\E_{\x}\ \bm{a}^{*\top}\bm{p}[F_j(\bar{\bT})-F_j(\bT^*)]G_j(\bar{\bT})^{\top}\|_{\infty}+\|\E_{\x}\ \bm{a}^{*\top}\bm{p}F_j(\bT^*)[G_j(\bar{\bT})-G_j(\bT^*)]^{\top}\|_{\infty}
\\
\leq &\ 
2\epsilon M^2K.
\end{align*}
}\noindent
Using a similar argument, we can have $\|\nabla_{ij}^2R(\bT^*,\bar{\bT})-\nabla_{ij}^2R(\bT^*,\bT^*)\|_{\infty}\leq 2\epsilon M^2K$ for any $i\neq j$ as well. Therefore, we verify Eq. (\ref{eq:b2}). Now, by using similar technique as in the proof of Proposition \ref{prop:mis4}, we can write
$-\nabla_{\bT}^2R(\bT^*,\bar{\bT})=\H+\delta\Delta$. Then, by the Woodbury matrix identity, we obtain
{\small
\[
-\nabla_{\bT}^2R(\bT^*,\bar{\bT})^{-1}=\H^{-1}-
\delta\left[\H^{-1}\Delta(\mathbb{I}+\delta\H^{-1}\Delta)^{-1}\H^{-1}\right],
\]
}\noindent
from which we can obtain the desired bound.
\end{proof}
Proposition~\ref{prop:mis5} shows that the bias relies on $\omega$ which is the difference between the misspecified distribution $\bm{p}^*$ and the true distribution $\bm{p}$. Given the misspecified model, $\omega$ would be fixed and hence the bias of the LUS estimator should not be neglected. Fortunately, in the case that $\bm{p}^*$ is close to $\bm{p}$, i.e., we can have a sufficiently small $\omega$, the difference between $\hat{\bT}_{LUS}$ and $\bT^*$ is also bounded in Proposition~\ref{prop:mis5}. 
%Otherwise, if $\bm{p}^*$ is far away from $\bm{p}$, it is also fine because estimating a parameter that is close to $\bT^*$ is of less interest either.

\clearpage

\section{Neural Net Structure Used in MNIST Data}
\label{appendix:neural}
We adopt one of the state-of-the-art convolutional neural net structures provided in the MatConvNet library \cite{vedaldi2015matconvnet} to study the performance of different sampling methods. The net structures for obtaining the rough estimate and the LUS/US sampling based estimator are shown in Figs. \ref{fig:net}(a) and \ref{fig:net}(b), respectively, where the net structure for obtaining the rough estimate is much simpler than the one used for LUS/US based estimator to save computation. Moreover, for the simpler neural net, we set the training epoch, i.e., the effective passes over data, to be 10. For the net used for the sampling based estimators, we set the epoch to be a larger value 20.

\begin{figure*}[!ht]
\centering
\subfigure[Structure for the pilot estimator]{
\includegraphics[scale=0.6]{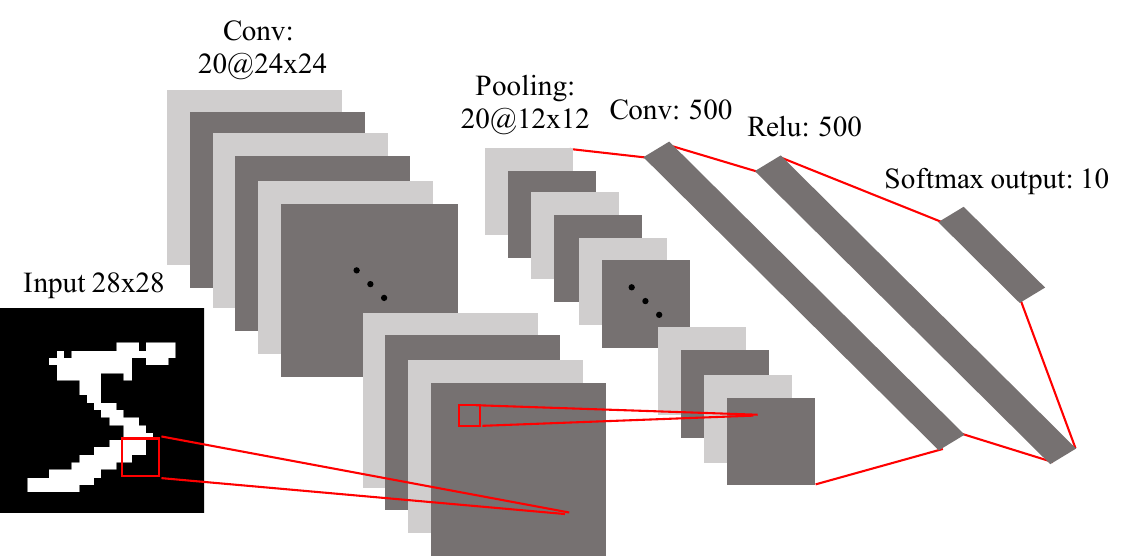}
}
\subfigure[Structure for the sampling methods]{
\includegraphics[scale=0.6]{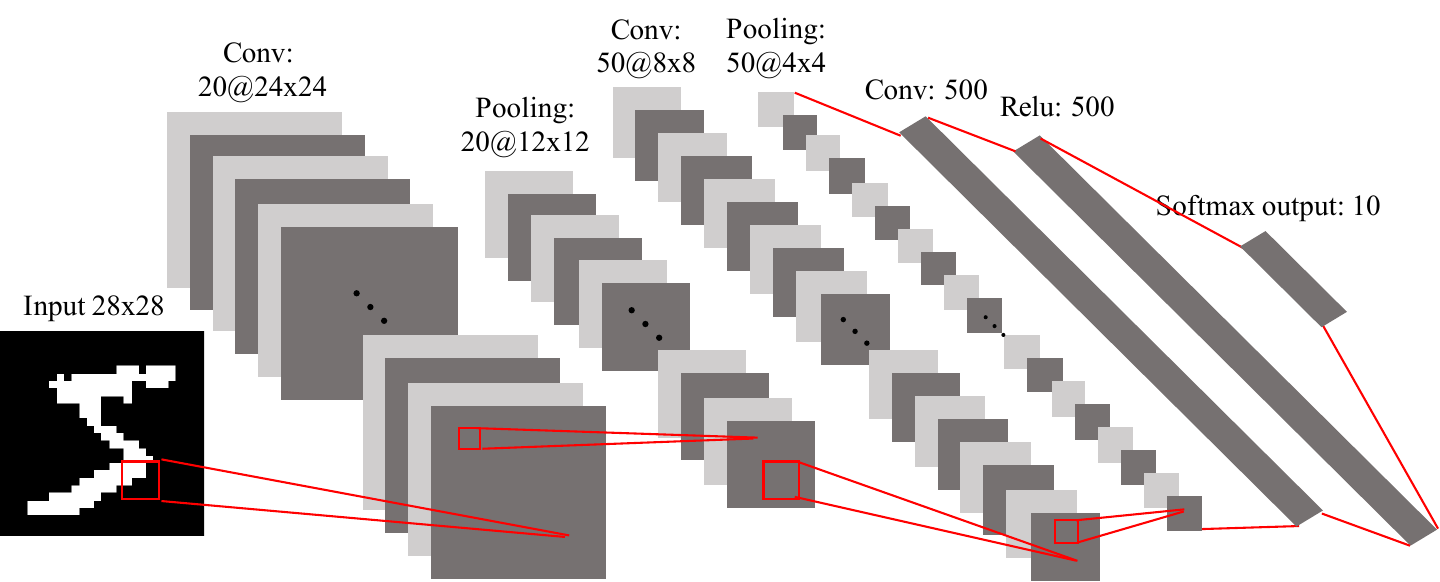}
}
\caption{Deep convolutional neural network structures used in the MNIST data.}
\label{fig:net}
\end{figure*}

\section{Simulations for Varying the Amount of Data Used for the Pilot Estimate}
\label{appendix:pilot_amount}

We provide simulation results to show how the performance of the LUS estimator changes when we vary the amount of data used for computing the pilot estimate. We keep the settings as the same as used in Section \ref{sec:simulation1} and Section \ref{sec:simulation2}. We change the portion of data used for the pilot estimate as $5\%$, $10\%$ and $20\%$, and evaluate the LUS method.

\begin{figure*}[t]
\centering
\subfigure[5\%]{
\includegraphics[scale=0.38]{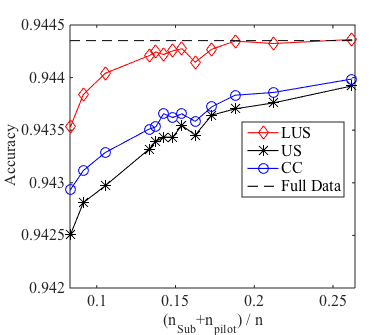}
}\hskip -0.15in
\subfigure[10\%]{
\includegraphics[scale=0.38]{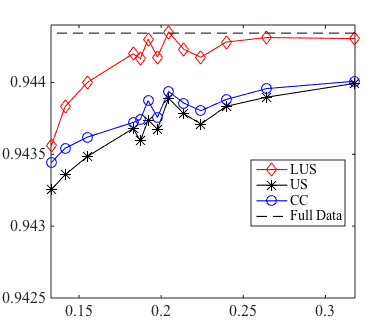}
}\hskip -0.1in
\subfigure[20\%]{
\includegraphics[scale=0.38]{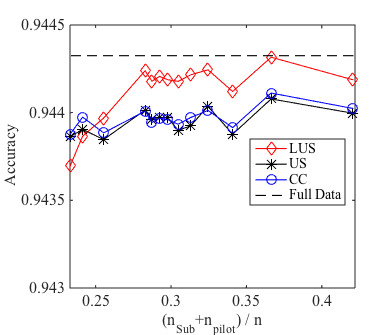}
}\hskip -0.1in
\caption{Performance of the LUS estimator when varying the amount of data for computing the pilot estimate under marginal imbalance.}
\label{fig:appendix_toy1}
\end{figure*}

\begin{figure*}[t]
\centering
\subfigure[5\%]{
\includegraphics[scale=0.38]{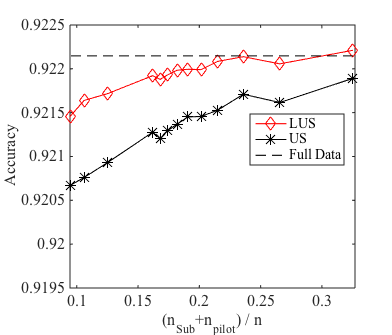}
}\hskip -0.15in
\subfigure[10\%]{
\includegraphics[scale=0.38]{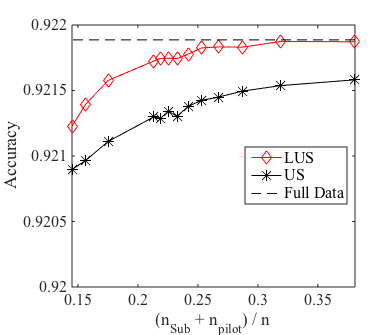}
}\hskip -0.1in
\subfigure[20\%]{
\includegraphics[scale=0.38]{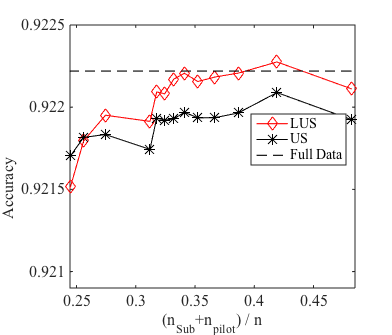}
}\hskip -0.1in
\caption{Performance of the LUS estimator when varying the amount of data for computing the pilot estimate under marginal balance.}
\label{fig:appendix_toy2}
\end{figure*}

Fig. \ref{fig:appendix_toy1} and Fig. \ref{fig:appendix_toy2} show the performance change of the LUS estimator in the simulations with marginal imbalance and marginal balance, respectively. In both of the simulations, we observe that the performance curves of the LUS method are similar that varying the amount of data for the pilot estimate does not affect the performance of the LUS estimator much and it generally needs a proportion of $\frac{n_{Sub}}{n}\approx 10\%$ to achieve the same performance as that of the full-sample based MLE; when we increase $n_{pilot}$, the uniform sampling and case-control sampling methods also sample more data points and the superiority of the LUS method is not as obvious as the settings with fewer data for the pilot estimate. The results suggest that we may use a small number of data points (e.g., $5\%$) to obtain a rough pilot estimate, which is sufficient to let the LUS method perform well.

\end{document}